\documentclass[12pt,onecolumn]{IEEEtran}

\usepackage{amsfonts}
\usepackage[sort]{cite} 
\usepackage{amsmath,amsthm} 
\usepackage{amssymb}  
\usepackage{graphicx}
\usepackage{epsfig}
\usepackage{epstopdf}
\usepackage{latexsym}
\usepackage{array}
\usepackage{amsfonts}
\usepackage{bm}
\usepackage{comment}
\usepackage{mathtools}
\usepackage{setspace}
\usepackage{etoolbox}
\usepackage{lipsum}
\usepackage{caption}
\usepackage{subcaption}
\usepackage{epstopdf}
\usepackage[colorlinks=true]{hyperref}

\def\IID{\mathop{\mathrm{i.i.d.}}}  
\def\diag{\mathop{\mathrm{diag}}}   
\def\trace{\mathop{\mathrm{trace}}} 
\def\rank{\mathop{\mathrm{rank}}}    
\def\op{\mathop{\mathrm{op}}}       
\def\rdf{\mathop{\mathrm{RDF}}}     
\def\awgn{\mathop{\mathrm{AWGN}}}   
\def\na{\mathop{\mathrm{na}}}       
\def\mse{\mathop{\mathrm{MSE}}}     
\def\mimo{\mathop{\mathrm{MIMO}}}     
\def\nrdf{\mathop{\mathrm{NRDF}}}   
\def\opta{\mathop{\mathrm{OPTA}}}   
\def\op{\mathop{\mathrm{op}}}       
\def\zd{\mathop{\mathrm{ZD}}}       
\def\sd{\mathop{\mathrm{SD}}}       
\def\sdusq{\mathop{\mathrm{SDUSQ}}} 
\def\ecdq{\mathop{\mathrm{ECDQ}}}   
\def\abs{\mathop{\mathrm{abs}}}      
\def\gp{\mathop{\mathrm{GP}}}       
\def\ed{\mathop{\mathrm{ED}}}       
\def\unif{\mathop{\mathrm{Unif}}}   
\def\kkt{\mathop{\mathrm{KKT}}}     
\def\rvs{\mathop{\mathrm{RVs}}}     
\def\rv{\mathop{\mathrm{RV}}}       
\def\ar{\mathop{\mathrm{AR}}}       
\def\vol{\mathop{\mathrm{Vol}}}      
\def\sdp{\mathop{\mathrm{SDP}}}     
\def\lmi{\mathop{\mathrm{LMI}}}     
\hypersetup{
     colorlinks = true,
     linkcolor = [rgb]{0,0,1},
     anchorcolor = [rgb]{0,0,1},
     citecolor = [rgb]{0.9,0.5,0},
     filecolor = [rgb]{0,0,1},
     pagecolor = [rgb]{1,1,0},
     urlcolor = [rgb]{0.9,0.5,0},
     bookmarks,
        bookmarksopen = true,
        bookmarksnumbered = true,
        breaklinks = true,
        linktocpage,
        pagebackref,
        colorlinks = true,
        linkcolor = [rgb]{0.2,0.6,0.2},
        urlcolor  = [rgb]{0,0,1},
        citecolor = [rgb]{0.9,0.5,0},
        anchorcolor = [rgb]{0.2,0.6,0.2},
        hyperindex = true,
        hyperfigures
}

\newtheorem{theorem}{Theorem}
\newtheorem{lemma}{Lemma}
\newtheorem{definition}{Definition}
\newtheorem{corollary}{Corollary}

\newtheorem{remark}{Remark}

\newtheorem{example}{Example}

\newtheorem{problem}{Problem}

\newcommand{\T}{^{\mbox{\tiny T}}} 
\newcommand{\G}{^{\mbox{\tiny G}}}

\newcommand{\noi}{\noindent}

\newcommand{\be}{\begin{equation}}
\newcommand{\ee}{\end{equation}}
\newcommand{\bea}{\begin{eqnarray}}
\newcommand{\eea}{\end{eqnarray}}
\newcommand{\bes}{\begin{eqnarray*}}
\newcommand{\ees}{\end{eqnarray*}}

\newcommand{\bfi}{\begin{figure}}
\newcommand{\bfit}{\begin{figure}[t]}
\newcommand{\bfib}{\begin{figure}[b]}
\newcommand{\bfih}{\begin{figure}[h]}
\newcommand{\bfip}{\begin{figure}[p]}
\newcommand{\efi}{\end{figure}}

\newcommand{\bi}{\begin{itemize}}
\newcommand{\ei}{\end{itemize}}
\newcommand{\ben}{\begin{enumerate}}
\newcommand{\een}{\end{enumerate}}

\newcommand{\bp}{\begin{problem}}
\newcommand{\ep}{\end{problem}}

\onehalfspace

\begin{document}

\sloppy

\title{Zero-Delay Rate Distortion via Filtering for Vector-Valued Gaussian Sources}


\author{{\IEEEauthorblockN{Photios A. Stavrou, Jan {\O}stergaard and Charalambos D. Charalambous
\thanks{Part of this work was presented at the Symposium on Information Theory and Signal Processing in the Benelux, Delft, Netherlands \cite{stavrou:2017} and at the IEEE Information Theory Workshop (ITW), Kaohsiung, Taiwan \cite{stavrou:2017itw}. This work has received funding from VILLUM FONDEN Young
Investigator Programme, under grant agreement No. 10095.}
\thanks{\IEEEauthorblockA{
P. A. Stavrou is with the Department of Information Science and Engineering, KTH Royal Institute of Technology, Sweden}, {\it email: fstavrou@kth.se}}
\thanks{\IEEEauthorblockA{J. {\O}stergaard is with the Department of Electronic Systems, Aalborg University, Denmark}, {\it email: jo@es.aau.dk}}
\thanks{\IEEEauthorblockA{C. D. Charalambous is with the Department of Electrical and Computer Engineering, University of Cyprus, Cyprus}, {\it email: chadcha@uc.ac.cy}}
}}}


\maketitle

\begin{abstract}
We deal with zero-delay source coding of a vector-valued Gauss-Markov source subject to a mean-squared error ($\mse$) fidelity criterion {characterized by the operational zero-delay vector-valued Gaussian rate distortion function ($\rdf$). We address this problem by considering the nonanticipative $\rdf$ ($\nrdf$) which is a lower bound to the causal optimal performance theoretically attainable ($\opta$) function (or simply causal $\rdf$) and operational zero-delay $\rdf$. We recall the realization that corresponds to the optimal ``test-channel'' of the Gaussian $\nrdf$, when considering a vector Gauss-Markov source subject to a $\mse$ distortion in the finite time horizon. Then, we introduce sufficient conditions to show existence of solution for this problem in the infinite time horizon (or asymptotic regime). For the asymptotic regime, we use the asymptotic characterization of the Gaussian $\nrdf$ to provide a new equivalent realization scheme with feedback which is characterized by a resource allocation (reverse-waterfilling) problem across the dimension of the vector source. We leverage the new realization to derive a predictive coding scheme via lattice quantization with subtractive dither and joint memoryless entropy coding. This coding scheme offers an upper bound to the operational zero-delay vector-valued Gaussian $\rdf$. When we use scalar quantization, then for $r$ active dimensions of the vector Gauss-Markov source the gap between the obtained lower and theoretical upper bounds is less than or equal to $0.254r + 1$ bits/vector. However, we further show that it is possible when we use vector quantization, and assume infinite dimensional Gauss-Markov sources to make the previous gap to be negligible, i.e., Gaussian $\nrdf$ approximates the operational zero-delay Gaussian $\rdf$. We also extend our results to vector-valued Gaussian sources of any finite memory under mild conditions}. Our theoretical framework is demonstrated with illustrative numerical experiments.
\end{abstract}

%
%
%
%
%
\section{Introduction}\label{sec:introduction}

\par Rate distortion theory describes the fundamental limits between the desired bitrate and the associated achievable distortion or vice versa, for a specific source and distortion measure \cite{berger:1971}. The source coders and decoders, which are able to get very close to the fundamental rate-distortion limits are generally computationally expensive, non-causal, and tends to impose long delays on the end-to-end processing of information. When source coding is to be part of a bigger infrastructure such as distributed data processing over sensor networks, networked control systems, etc., there will often be strict requirements on the tolerable delay and system complexity. {This necessitates real-time communication between the systems involved whereas delays play a critical role on the performance or even the stability of these systems.}

\par To achieve near instantaneous encoding and decoding, it is necessary that the source encoder and decoder are causal \cite{neuhoff:1982}. Unfortunately, causality comes with a price. In particular, it was shown in \cite{derpich:2012} that imposing causality on the coder results in an increase in the bitrate due to the quantizer's space-filling loss and reduced de-noising capabilities due to causal filtering at the decoder. If zero-delay is furthermore imposed, there will be an additional increase in the bitrate due to having a finite (and often small) alphabet in the entropy coder~\cite{derpich:2012}.

\par In applications where both instantaneous encoding and
decoding are required, it is common to use the term zero-delay source coding \cite{linder:2001}. Zero-delay source coding is particularly relevant for networked control systems, where an unstable plant is to be stabilized via a communication channel. At each time step, the feedback signal of the plant needs to be encoded, transmitted over a channel, decoded, and reproduced at the controller's side. Some indicative works on zero-delay source coding for control systems can be found, for instance, in \cite{witsenhausen1979,tatikonda:2004,nair:2007,silva:2011,yuksel:2014,tanaka:2016,silva:2016,khina:2017itw}.

\par In the field of information theory, there is a tradition to establish achievability of a certain rate-distortion performance by showing a construction based on random codebooks, which requires asymptotically large source vector dimensions \cite{cover-thomas2006}. However, in the case of zero-delay source coding, the random coding based technique is often not applicable. Indeed, the optimal rate-distortion performance for zero-delay source coding{, hereinafter called zero-delay rate distortion,} is hard to establish and is, for example, not known for the case of general Gaussian sources subject to a mean squared error ($\mse$) distortion, whereas the non-causal {classical} rate distortion function ($\rdf$) is, in general, known. {To overcome the computational complexity of the zero-delay $\rdf$, there has been a turn in studying variants of classical $\rdf$ that are lower bounds to the zero-delay $\rdf$. One such variant is the so-called nonanticipative $\rdf$ ($\nrdf$) also found as nonanticipatory $\epsilon-$entropy and sequential $\rdf$ in the literature. {The $\nrdf$ was first introduced in \cite{gorbunov-pinsker1973} and extensively  analysed for Gauss-Markov sources  in \cite{gorbunov-pinsker1974}.  In \cite[Theorem 5]{gorbunov-pinsker1974}, the authors derived a partial characterization (because certain parameters are not found) for $\nrdf$ for time-varying vector-valued Gauss-Markov sources with square-error distortion function,  by providing a parametric realization of the test channel conditional distribution of the reproduction process, that is first-order Markov with respect to source symbols and depends only on the previous reproduction symbol. Moreover, in \cite[Examples 1, 2]{gorbunov-pinsker1974} the authors derive the complete characterization of the $\nrdf$, for time-varying and stationary scalar-valued Gaussian first-order autoregressive ($\ar(1)$) processes, with pointwise or per-letter mean squared-error ($\mse$) distortion {fidelity} and gave the expression of finite-time $\nrdf$ in terms of a reverse-waterfilling at each time instant and the corresponding expression in the asymptotic regime.} Tatikonda {\it et al.}  in \cite{tatikonda:2004} leverage the results of \cite[Examples 1, 2]{gorbunov-pinsker1974} and applied the asymptotic $\nrdf$ to compute the Gaussian $\nrdf$ for time-invariant scalar-valued Gaussian $\ar(1)$ sources with an asymptotic $\mse$ distortion constraint. {In addition, they gave a parametric expression of the $\nrdf$ for time-invariant vector-valued Gauss-Markov
sources, that is described by a reverse-waterfilling algorithm which is unfortunately suboptimal (the suboptimality is demonstrated via a counterexample in \cite{stavrou-tanaka-tatikonda:2018}).} {It should be noted that in \cite{tatikonda:2004} and also \cite{stavrou-tanaka-tatikonda:2018} the authors do not attempt to identify the parameters of the realization given in \cite[Theorem 5]{gorbunov-pinsker1974}. {Derpich and ${\O}$stergaard in \cite{derpich:2012} considered  variants of $\nrdf$ for stationary scalar-valued Gaussian autoregressive models of any order with pointwise $\mse$ distortion {fidelity} and computed the asymptotic expression of the Gaussian $\nrdf$ for stationary scalar-valued Gaussian $\ar (1)$ sources which was first derived in \cite[Equation (1.43)]{gorbunov-pinsker1974} using alternative methods.   Tanaka {\it et al.} in \cite{tanaka:2017} revisited the finite-time $\nrdf$ for vector-valued Gauss-Markov sources  with pointwise $\mse$ distortion {fidelity} following the line of work of \cite{tatikonda:2004} and showed that the resulting optimization problem is semidefinite representable, thus, it can be solved numerically. However, none of the previously discussed works, i.e., \cite{tatikonda:2004,derpich:2012,tanaka:2017,stavrou-tanaka-tatikonda:2018} provide a realization of the test channel conditional distribution that achieves the $\nrdf$ or attempt to identify the parameters in the realization given in \cite[Theorem 5]{gorbunov-pinsker1974}.} Stavrou {\it et al.} in \cite{stavrou-charalambous-charalambous2016} considered the $\nrdf$ of the time-varying vector-valued Gauss-Markov source under a total $\mse$ distortion {fidelity}, and gave a sub-optimal realization of the design coefficient in the parametric realization given in \cite[Theorem 5]{gorbunov-pinsker1974}. Further, in \cite[Theorem 2]{stavrou-charalambous-charalambous2016} the computation of the finite-time $\nrdf$ via  a dynamic reverse-waterfilling algorithm is not optimal (this is explained in \cite{stavrou-charalambous-charalambous:2018lcss}).  Recently, in  \cite{stavrou-charalambous-charalambous-loyka2018siam} the authors computed the finite-time $\nrdf$ for vector-valued Gauss-Markov sources subject to a total and per-letter $\mse$ distortion constraint, using convex optimization techniques and gave a parametric solution via a dynamic reverse-waterfilling algorithm, that identifies the parametric realization given in \cite[Theorem 5]{gorbunov-pinsker1974}. The results obtained in \cite{stavrou-charalambous-charalambous-loyka2018siam} did not consider the asymptotic regime.}}
\par {The signal processing approaches to source coding can roughly be classified into transform coding \cite{goyal:2001,vetterli:2001}, filterbanks~\cite{vaidyanathan:1993,gosse:1997}, and predictive coders~\cite{atal:1970,makhoul:1975,gersho:1991,vaidyanathan:1993}. A transform can be put on a matrix form, which is multiplied on the signal vector. Clearly, this operation is only causal if the matrix is lower triangular (when multiplied on the left hand side of the signal vector). Low delay filters have been considered in \cite{nayebi:1993}, and zero-delay filtering in \cite{charalambous-stavrou-ahmed2014,stavrou-charalambous-charalambous2016}.  Predictive coders usually operates directly on the time-domain signal, and can easily be made causal (and of zero-delay) by simply only making use of the current and past samples of the source signal. Recently, it has been shown that the causal $\rdf$ of a stationary colored scalar Gaussian process can be achieved by \emph{causal} prediction and noise-shaping \cite{derpich:2012}.\footnote{This result parallels that of \cite{zamir:2008}, where it was shown that the non-causal $\rdf$ of a  stationary colored scalar Gaussian process under $\mse$ can be achieved by (non-causal) prediction.}

\par {In this paper, we deal with zero-delay source coding of a vector-valued Gauss-Markov source expressed in state space form subject to a $\mse$ distortion constraint. We recall the complete characterization of the finite-time $\nrdf$ for time-varying Gauss-Markov sources subject to a total $\mse$ distortion developed for scalar-valued sources in \cite{stavrou-charalambous-charalambous:2018lcss} and for vector-valued sources in \cite{stavrou-charalambous-charalambous-loyka2018siam} to obtain the following results.
\begin{itemize}
\item[(1)] Sufficient conditions to ensure by construction existence of solution of the per unit time asymptotic limit of the finite-time Gaussian $\nrdf$. The asymptotic Gaussian $\nrdf$ provides a lower bound to the operational zero-delay vector-valued Gaussian $\rdf$.
\item[(2)] A new feedback realization scheme that corresponds to the optimal test channel of the asymptotic Gaussian $\nrdf$. This scheme is characterized by a resource allocation problem across the dimension of the source. 
\item[(3)] A coding scheme based on predictive coding  which is applied to this feedback realization scheme using scalar or vector quantization and joint entropy coding separately across every dimension of the vector-valued Gauss-Markov source. This scheme provides an achievable (upper) bound to the operational zero-delay vector-valued Gaussian $\rdf$.
\item[(4)] Several numerical examples that demonstrate our theoretical framework. These examples take into account both stable and unstable Gaussian sources.
\end{itemize}
In addition to the previous main results, we explain how our scheme can be generalized to vector-valued Gauss-Markov processes of any finite order.}\\
{The new feedback realization scheme has a Kalman filter in the feedback loop. The feedback loop serves two purposes; if the Gaussian source is unstable then the filter with the help of the feedback loop tracks it while the estimation error converges, and it removes most of the source redundancy along the temporal direction by means of closed-loop vector prediction. On the other hand, the feed-forward path transforms the residual (innovations) vector Gaussian source into a new vector source, which has independent spatial components, and thereby can be efficiently encoded by applying for example scalar quantization and joint entropy coding separately across each dimension of the vector.  Our construction makes use of simple building blocks such as non-singular joint diagonalizers (KLT matrices), diagonal scaling matrices, Kalman filters, and scalar (or lattice) quantization. It also demonstrates the resource allocation of the source signals depending on the data rate budget. This means that our scheme demonstrates which dimensions are active when the reverse-waterfilling kicks in. This issue and the complete machinery to obtain theoretical lower and upper bounds as well as the operational rates to the zero-delay Gaussian $\rdf$ is not demonstrated in the recent works of \cite{stavrou:2017,stavrou:2017itw}.}
\par {Our coding results demonstrate that when we use scalar quantization, the gap between the obtained lower and theoretical upper bounds is less than or equal to $0.254r + 1$ bits/vector where $r$ denotes the number of active dimensions of the vector-valued Gauss-Markov source. Moreover, at high rates our simulation experiments demonstrate that the gap between the lower bound and the operational rates mitigates to approximately $0.254{r}$ bits/vector. For vector quantization, we show that in the limit of asymptotically large vector dimensions, it is possible for the causal and zero-delay $\rdf$ to coincide with the Gaussian $\nrdf$.}
\par { It should be noted that our realization scheme can be paralleled to the work developed in \cite[Chapter 11]{matveev-savkin:2009} (see also the references therein) where various (possibly partially observable) source signals are communicated via an observer or controller over parallel Gaussian channels with spatially independent delays. Compared to that framework we investigate perfect prediction in the sense that we do not take into account issues like data dropouts or delays within the parallel channels or even conditions for stability of the estimator. Potentially, one can leverage our framework to investigate similar problems to \cite[Chapter 11]{matveev-savkin:2009}. }  

\par This paper is structured as follows. {In \S \ref{sec:problem_formulation} we characterize the zero-delay source coding problem for vector-valued Gauss-Markov sources subject to an asymptotic $\mse$ distortion constraint in terms of zero-delay Gaussian $\rdf$. In \S\ref{sec:lower_bounds_zero_delay_codes} we give known lower bounds to zero-delay Gaussian $\rdf$ using general Gaussian sources subject to a $\mse$ distortion whereas in \S\ref{subsec:specific_prob} we concentrate on the $\nrdf$ of vector-valued Gauss-Markov source. In \S\ref{sec:new_realization_scheme} we show existence of solution to the asymptotic Gaussian $\nrdf$ and we provide a new feedback realization scheme that corresponds to the asymptotic test-channel of this problem. \S\ref{sec:main_results} derives upper bounds to the zero-delay Gaussian $\rdf$ in terms of the Gaussian $\nrdf$ using scalar and vector quantization with memoryless entropy coding. In \S\ref{sect:examples} we demonstrate our theoretical framework via several numerical experiments. We draw conclusions and discuss future directions in \S\ref{sec:conclusions}.}

%
%
%
%
\paragraph*{\bf Notation}  $\mathbb{R}$  denotes the set of real numbers,  $\mathbb{Z}$  the set of integers,  $\mathbb{N}_0$ the set of natural numbers including zero, and $\mathbb{N}^n_0\triangleq\{0,\ldots,n\}, n \in \mathbb{N}_0$. {Let ${\cal X}$ be a finite dimensional Euclidean space, and ${\cal B}({\cal X})$ be the Borel $\sigma$-algebra on ${\cal X}$.} A random variable ($\rv$) defined on some probability space $(\Omega, {\cal F}, {\mathbb P})$ is a map $X: \Omega \longmapsto {\cal X}$. We denote a sequence of $\rvs$ by ${\bf x}_r^t \triangleq ({\bf x}_r, {\bf x}_{r+1}, \ldots,{\bf x}_t), (r, t) \in {\mathbb Z}\times {\mathbb Z}, t\geq r$, and their values by ${x}_r^t \in  {\cal X}_r^t \triangleq \times_{k=r}^t {\cal X}_k$, with ${\cal X}_k={\cal X}$, for simplicity. If $r=-\infty$ and $t=-1$, we use the notation ${\bf x}_{-\infty}^{-1}={\bf x}^{-1}$, and if $r=0$,  we use the notation ${\bf x}_0^t = {\bf x}^t$. The distribution of the $\rv$ ${\bf x}$ on ${\cal X}$ is denoted by ${\bf P}_{\bf x}\equiv{\bf P}(dx)$. The conditional distribution of ${\rv}$ ${\bf y}$ given ${\bf x}=x$ is denoted by ${\bf P}_{{\bf y}|{\bf x}}\equiv{\bf P}(dy|x)$. The transpose of a matrix or vector ${K}$ is denoted by ${K}{\T}$. The covariance of a random vector $K$ is denoted by $\Sigma_K$. {For a square matrix $K\in \mathbb{R}^{p\times p}$, we denote the diagonal by $\diag(\mu_{{K},i})$, where $\mu_{{K},i}$ denotes the $i$th eigenvalue of matrix $K$, its  determinant by $|K|$, its trace by $\trace\{K\}$, and its rank by $\rank(K)$. We denote by $K \succ 0$ (respectively, $K \succeq 0$) a symmetric positive-definite matrix (respectively, symmetric positive-semidefinite matrix). The statement $\Sigma_K\succeq\Sigma_{K^{\prime}}$ means that $\Sigma_K-\Sigma_{K^{\prime}}$ is positive semidefinite. We denote identity matrix by $I$. We denote by $(\cdot)\G$ any $\rv$ or a vector that is Gaussian distributed. We denote by $\mathbb{H}(\cdot)$ the discrete entropy and by $h(\cdot)$ the differential entropy. $\mathbb{D}(P||Q)$  denotes the relative entropy of probability distribution $P$ with respect to probability distribution $Q$. We denote by $\log\text{abs}(\cdot)$ the absolute value of a quantity in the logarithm.}

%
%
%
%
%
\section{Problem Statement}\label{sec:problem_formulation}

\par In this paper we consider the zero-delay source coding setting illustrated in Fig.~\ref{fig:zero_delay_system}. In this setting, the $\mathbb{R}^p$-valued Gaussian source is governed by the following discrete-time linear time-invariant Gaussian state-space model
\begin{align}
{\bf x}_{t+1}=A{\bf x}_{t}+B{\bf w}_t,~{\bf x}_0=\bar{x},~t\in\mathbb{N}_0,\label{state_space_model}
\end{align}
\noi where $A\in\mathbb{R}^{p\times{p}}$ and $B\in\mathbb{R}^{p\times{q}}$ are deterministic matrices, ${\bf x}_0\in\mathbb{R}^p\sim{\cal N}(0;\Sigma_{{\bf x}_0})$ is the initial state, ${\bf w}_t\in\mathbb{R}^q\sim{\cal N}(0;\Sigma_{\bf w})$, $\Sigma_{\bf w}=I$, is an $\IID$ Gaussian sequence, independent of ${\bf x}_0$.    
\par The system operates as follows. At every time step $t\in\mathbb{N}_0$, the {\it encoder} observes the vector source ${\bf x}^t$ and produces a single binary codeword ${\bf z}_t$ from a predefined set of codewords ${\cal Z}_t$ of at most countable number of codewords. Since the source is random, ${\bf z}_t$ and its length ${\bf l}_t$ (in bits) are random variables. Upon receiving ${\bf z}^t$, the {\it decoder} produces an estimate ${\bf y}_t$ of the source sample {${\bf x}_t$, under the assumption that ${\bf y}^{t-1}$ is already reproduced}. We assume that both the encoder and decoder process information without delay.
\begin{figure}[htp]
\centering
\includegraphics[width=\columnwidth]{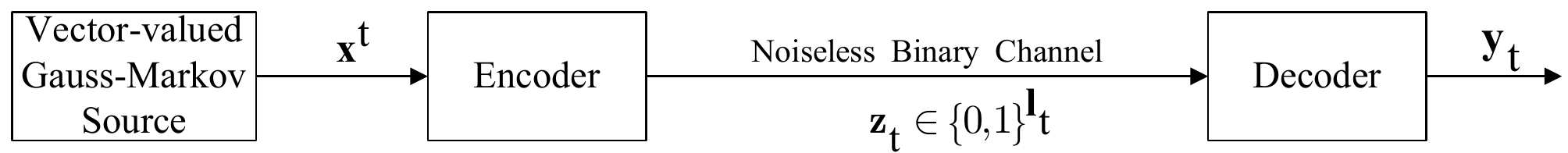}
\caption{A zero-delay source coding scenario using variable-length binary codewords.}\label{fig:zero_delay_system}
\end{figure}

\par The analysis of the noiseless digital channel is restricted to the class of instantaneous variable-length binary codes ${\bf z}_t$. The countable set of all codewords (codebook) ${\cal Z}_t$ is time-varying to allow the binary representation ${\bf z}_t$ to be an arbitrarily long sequence. 

\paragraph*{Zero-delay source coding} Formally, the zero-delay source coding problem of Fig. \ref{fig:zero_delay_system} can be explained as follows. Define the input and output alphabet of the noiseless digital channel by ${\cal M}=\{1,2,\ldots, {M}\}$ where ${M}=\max_t|{\cal Z}_t|$ (possibly infinite). The elements in ${\cal M}$ enumerate the codewords of ${\cal Z}_t$. The encoder is specified by the sequence of measurable functions $\{f_t:~t\in\mathbb{N}_0\}$ with $f_t:{\cal M}^{t-1}\times{\cal X}^{t}\rightarrow{\cal M}$. At time $t\in\mathbb{N}_0$, the output of the encoder is a message ${z}_t=f_t({z}^{t-1},{x}^t)$ with $z_0=f_0(x_0)$ which is transmitted through a noiseless channel to the decoder. The decoder is specified by the sequence of measurable functions $\{g_t:~t\in\mathbb{N}_0\}$ with $g_t:{\cal M}^{t}\rightarrow{\cal Y}_t$. For each $t\in\mathbb{N}_0$, the decoder generates ${y}_t=g_t(z^t)$ with ${y}_0=g_0(z_0)$ assuming $y^{t-1}$ is already generated.

\paragraph*{Asymptotic $\mse$ distortion constraint} The design in Fig. \ref{fig:zero_delay_system} is required to yield an asymptotic average distortion $\limsup_{n\longrightarrow\infty}\frac{1}{n+1}\mathbb{E}\{d({\bf x}^n,{\bf y}^n)\}\leq{D}$, where $D>0$ is the pre-specified distortion level, $d({\bf x}^n,{\bf y}^n)\triangleq\sum_{t=0}^n||{\bf x}_t-{\bf y}_t||_{2}^2$. For the asymptotic regime, the objective is to minimize the expected average codeword length, i.e., the total number of bits received by the decoder at the time it reproduces $\{{\bf y}_t:~t\in\mathbb{N}_0\}$, denoted by $\limsup_{n\longrightarrow\infty}\frac{1}{n+1}\sum_{t=0}^n\mathbb{E}({\bf l}_t)$, over all measurable zero-delay encoding and decoding functions $\{(f_t,g_t):~t\in\mathbb{N}_0\}$. {We denote by ${\bf L}_n\triangleq\sum_{t=0}^n{\bf l}_t$ the accumulated number of bits received by the decoder at the time it reproduces the estimate ${\bf y}_n$.}
{
\begin{problem}(Zero-delay vector-valued Gaussian $\rdf$)\label{problem1}{\ \\}
The previous design requirements are formally cast by the following optimization problem:
\begin{align}
R^{\op}_{\zd}(D)&\triangleq\inf_{\substack{z_t=f_t(z^{t-1},x^t),~t\in\mathbb{N}_0\\~y_t=g_t(z^t)}}\limsup_{n\longrightarrow\infty}\frac{1}{n+1}\mathbb{E}({\bf L}_n).\label{def:operational_zero_delay}\\
&\text{s. t.}~~\limsup_{n\longrightarrow\infty}\frac{1}{n+1}\mathbb{E}\left\{\sum_{t=0}^n||{\bf x}_t-{\bf y}_t||_2^2\right\}\leq{D}\nonumber
\end{align}
We refer to  \eqref{def:operational_zero_delay} as the {\it operational zero-delay Gaussian $\rdf$}.
\end{problem}
}
\par Unfortunately, the solution of Problem \ref{problem1} is very hard to find because it is defined over all operational codes. For this reason, in the next section we introduce a lower bound to this problem which is defined based on information theoretic quantities.

%
%
%

\section{Lower Bounds on {Problem \ref{problem1}}}\label{sec:lower_bounds_zero_delay_codes}

\par In this section, we present {known lower bounds to the operational zero-delay Gaussian $\rdf$ of Problem \ref{problem1}. To do so, first we formally introduce the definitions of causal source coder, causal optimal performance theoretically attainable ($\opta$) function and $\nrdf$ {(and its relation to \cite{gorbunov-pinsker1973})} assuming  general Gaussian sources (although the same bounds apply to non-Gaussian sources) subject to an asymptotic $\mse$ distortion constraint. Then, we concentrate on the specific lower bound to Problem \ref{problem1} investigated in this paper.}

\paragraph*{Causal $\opta$ function} {In general source coding, a source encoder-decoder ($\ed$) pair encodes a source $\{{\bf x}_t:~t\in\mathbb{N}_0\}$ distributed according to ${\bf P}_{{\bf x}^\infty}\equiv{\bf P}(dx^\infty)\triangleq\otimes_{t=0}^\infty{\bf P}(dx_t|x^{t-1})$ with ${\bf P}(dx_0|x^{-1})={\bf P}(dx_0)$, into binary representations from which the estimate $\{{\bf y}_t:~t\in\mathbb{N}_0\}$ of $\{{\bf x}_t:~t\in\mathbb{N}_0\}$ is generated. The end-to-end effect of any $\ed$ pair is captured by a sequence of reproduction functions $\{{f}_t:~t\in\mathbb{N}_0\}$ such that
\begin{align*}
y_t=f_t(x^{\infty}),~t\in\mathbb{N}_0.
\end{align*}
Following \cite{neuhoff:1982} the $\ed$ pair is called causal if the following definition holds.
\begin{definition}({Causal reproduction coder}){\ \\}\label{def:causal_reproduction_coder}
A sequence of reproduction coders $\{f_t:~t\in\mathbb{N}_0\}$, is called causal if for each $t$,
\begin{equation}
f_t(x^\infty)=f_t(\tilde{x}^\infty)~~\text{whenever}~~x^t={\tilde x}^t,~t\in\mathbb{N}_0.\label{causal_reproduction_coder:eq.1}
\end{equation}
A causal source code is induced by a causal reproduction coder.
\end{definition}
}
\par {Next, we give the definition of the causal $\opta$ function \cite{neuhoff:1982}.
\begin{definition}(Causal $\opta$ function)\label{definition:causal_opta}{\ \\} 
For $D>0$, the minimum rate achievable by any causal source code  with distortion not exceeding $D$ is given by the causal $\opta$ function defined by 
\begin{align}
R_{c}^{\op}(D)&\triangleq\inf_{\substack{y_t:~y_t=f_t(x^t),~{t}\in\mathbb{N}_0\\~\{f_t:~t\in\mathbb{N}_0\}~\text{is~causal}}}\limsup_{n\longrightarrow\infty}\frac{1}{n+1}{\mathbb{E}}\left({\bf L}_n\right).\label{causal_opta_function}\\
&\text{s. t.}~~\limsup_{n\longrightarrow\infty}\frac{1}{n+1}\mathbb{E}\left\{\sum_{t=0}^n||{\bf x}_t-{\bf y}_t||_2^2\right\}\leq{D}\nonumber
\end{align}
\end{definition}}
\vspace{-0.15cm}
\paragraph*{$\nrdf$} We consider {a source} that randomly generates sequences ${\bf x}_t=x_t \in {\cal X}_t, t\in\mathbb{N}_0^n$, that we wish to reproduce or reconstruct  by  ${\bf y}_t=y_t \in {\cal Y}_t,~t\in\mathbb{N}_0^n$, subject to a $\mse$ distortion constraint defined by $d_{0,n}(x^n, y^n)\triangleq\sum_{t=0}^n ||{x}_t-{y}_t||_2^2$.\\
\noi{\bf Source.} The source distribution satisfies conditional independence  
\begin{align}
{\bf P}_{{\bf x}_t|{\bf x}^{t-1}, {\bf y}^{t-1}}\triangleq {\bf P}(dx_t|x^{t-1}), \ \ t\in\mathbb{N}_0^n. \label{source;distributions}
\end{align}
Since no initial information is assumed, the distribution at $t=0$  is ${\bf P}(dx_0)$. Also, by Bayes' rule we obtain
${\bf P}_{{\bf x}^n}\equiv{\bf P}(dx^n)\triangleq\otimes_{t=0}^n{\bf P}(dx_t|x^{t-1})$. Note that for model (\ref{state_space_model}), (\ref{source;distributions}) implies that ${\bf w}_t$ is independent of the past reproductions ${\bf y}^{t-1}$.

\noi{\bf Reproduction or ``test-channel''.} {Suppose the reproduction ${\bf y}^n=y^n, n \in {\mathbb N}_0$ of $x^n $ is randomly generated, according to the collection of conditional distributions, known as test-channels, by 
\begin{align}
{\bf P}_{{\bf y}_t|{\bf y}^{t-1}, {\bf x}^{t}}\triangleq {\bf P}(dy_t|y^{t-1},x^t), \ \ t\in\mathbb{N}_0^n. \label{reproduction:distributions}
\end{align}
{At $t=0$, no initial state information is assumed, hence  ${\bf P}(dy_0|y^{-1},x^0)={\bf P}(dy_0|x_0)$.} From  \cite[Remark 1]{charalambous-stavrou2016}, we know that the conditional distributions ${\bf P}(dy_t|y^{t-1},x^t)$ in  \eqref{reproduction:distributions}, uniquely define the family of conditional distributions on ${\cal Y}^{n}$ parametrized by $x^n \in {\cal X}^{n}$, given by
\begin{align}
\overrightarrow{\bf Q}(dy^n|x^n)\triangleq\otimes_{t=0}^n{\bf P}(dy_t|y^{t-1},x^t),\nonumber
\end{align}
and vice-versa. By \eqref{source;distributions} and \eqref{reproduction:distributions}, we can uniquely define the joint distribution {of $\{({\bf x}^n,{\bf y}^n):~t\in\mathbb{N}_0^n\}$ by
\begin{align}
{\bf P}_{{\bf x}^n, {\bf y}^n}^{\overrightarrow{\bf Q}}(dx^n,dy^n)={\bf P}(dx^n)\otimes \overrightarrow{\bf Q}(dy^n|x^n).\label{joint_distribution}
\end{align}
In addition, from \eqref{joint_distribution}, we can uniquely define the ${\cal Y}^{n}-$marginal distribution by 
\begin{align}
{\bf P}_{{\bf y}^n}^{\overrightarrow{\bf Q}}(dy^n)\triangleq \int_{{\cal X}^n}{\bf P}(dx^n)\otimes \overrightarrow{\bf Q}(dy^n|x^n),\nonumber
\end{align}
}
and the conditional distributions ${\bf P}_{{\bf y}_t|{\bf y}^{t-1}}^{\overrightarrow{\bf Q}}$,~$t\in\mathbb{N}_0^n$.}
\par {Given the above construction of distributions, {we introduce the information measure using relative entropy} as follows: 
\begin{subequations}
\begin{align}
&I^{\overrightarrow{\bf Q}}({\bf x}^n;{\bf y}^n) \stackrel{(a)}\triangleq {\mathbb D}({\bf P}_{{\bf x}^n,{\bf y}^n}^{\overrightarrow{\bf Q}}|| {\bf P}_{{\bf y}^n}^{\overrightarrow{\bf Q}}\times {\bf P}_{{\bf x}^n}) \in [0,\infty] \nonumber \\
&\stackrel{(b)}{=}\int_{{\cal X}^{n} \times {\cal Y}^{n}} \log \left( \frac{\overrightarrow{\bf Q}(\cdot|x^n)}{{\bf P}_{{\bf y}^n}^{\overrightarrow{\bf Q}}(\cdot)}(y^n)\right){\bf P}_{{\bf x}^n, {\bf y}^n}^{\overrightarrow{\bf Q}}(dx^n,dy^n)\label{eqdi5}
\end{align}
\begin{align}
&\stackrel{(c)}{=} \sum_{t=0}^n \mathbb{E}\left\{\log \left( \frac{{\bf P}(\cdot|{\bf y}^{t-1},{\bf x}^t)}{{\bf P}_{{\bf y_t|y^{t-1}}}^{\overrightarrow{\bf Q}}(\cdot|{\bf y}^{t-1})}({\bf y}_t)\right)\right\}\\
&\stackrel{(d)}{=}\sum_{t=0}^n I({\bf x}^t; {\bf y}_t|{\bf y}^{t-1}),  \label{eq_s}
\end{align} 
\end{subequations}
where $(a)$ follows by definition of relative entropy; 
$(b)$ is due to the Radon-Nikodym derivative theorem \cite[Appendix A.C]{charalambous-stavrou2016}; $(c)$ is due to chain rule of relative entropy;  $(d)$ follows by definition. Often, we use either  \eqref{eqdi5} or \eqref{eq_s}. It should be remarked that since \eqref{source;distributions}   and \eqref{reproduction:distributions} hold, then  (\ref{eq_s}) is a special case of directed information from ${\bf x}^n$ to ${\bf y}^n$ (see \cite{massey1990}).}

\par {Next, we formally define the Gaussian $\nrdf$ subject to a $\mse$ distortion. Recall that the following definition was announced in \cite{gorbunov-pinsker1973} for general distortion functions (including $\mse$ distortions) and \cite{gorbunov-pinsker1974} for pointwise $\mse$ distortion functions.  
\begin{definition}(Asymptotic Gaussian {$\nrdf$} subject to a $\mse$ distortion)\label{def:nonanticipative_rdf}{\ \\}
For the fixed Gaussian source of \eqref{source;distributions}, and a $\mse$ distortion the following holds.\\
(1) The finite-time $\nrdf$ is defined by 
\begin{eqnarray}
{R}^{\na}_{0,n}(D) \triangleq \inf_{\substack{{\bf P}(dy_t|y^{t-1},x^t):~t\in\mathbb{N}_0^n\\~\frac{1}{n+1}\mathbb{E}\left\{\sum_{t=0}^n||{\bf x}_t-{\bf y}_t||_2^2\right\}\leq{D}}}\frac{1}{n+1} I^{\overrightarrow{\bf Q}}({\bf x}^n;{\bf y}^n),\label{finite_time_nrdf}
\end{eqnarray}
{assuming the infimum is achieved in the set.}\\
(2) The asymptotic limit of \eqref{finite_time_nrdf} is defined by 
\begin{align}
{R}^{\na}&(D)=\lim_{n\longrightarrow\infty}{R}^{\na}_{0,n}(D),\label{infinite_time_nrdf}
\end{align}
assuming the infimum is achieved in the set and the limit exists and it is finite.
\end{definition}
}
\par {If one interchanges $\liminf$ to $\inf\lim$ in \eqref{infinite_time_nrdf}, then an upper bound to $R^{\na}(D)$ is obtained, defined as follows: 
\begin{align}
\widehat{R}^{\na}(D)&\triangleq\inf_{\overrightarrow{\bf Q}(dy^\infty|x^\infty)} \lim_{n\longrightarrow\infty}  \frac{1}{n+1} I^{\overrightarrow{\bf Q}}    ({\bf x}^n;{\bf y}^n)\label{infinite_time_nrdf_upper_bound} \\
&\text{s. t}\lim_{n\longrightarrow\infty} \frac{1}{n+1}\mathbb{E}\left\{\sum_{t=0}^n||{\bf x}_t-{\bf y}_t||_2^2 \right\} \leq D \nonumber
\end{align}
where $\overrightarrow{\bf Q}(dy^\infty|x^\infty)$ denotes the sequence of conditional probability distributions $\{{\bf P}(dy_t|y^{t-1},x^t):~t\in\mathbb{N}_0\}$.
\par {Next, we discuss some properties of the $\nrdf$ that can be extracted from different references.  First, it can be shown that the  optimization problem  \eqref{finite_time_nrdf}, in contrast to the one of \eqref{def:operational_zero_delay},  is convex with respect to the test channel,  for $D \in[D_{\min}, D_{\max}]\subseteq[0, \infty]$. Moreover, under mild conditions (given in \cite[Theorem 15]{charalambous-stavrou2016}), when the source is not necessarily Gaussian, the infimum is achieved and the $\nrdf$ is finite. By the structural properties of the test channel derived in { \cite[Theorem 1]{stavrou-charalambous-charalambous2016},} if the source is first-order Markov, i.e., with distribution ${\bf P}(dx_t|x_{t-1}), t\in\mathbb{N}_0^n$, the test channel distribution is of the form ${\bf P}(dy_t| y^{t-1},x_t),~ t\in\mathbb{N}_0^n$. Finally, combining this structural result, with  \cite[Theorem 1.8.6]{ihara1993}, it can be shown that if ${\bf x}^n$ is  Gaussian  then a jointly Gaussian  process $\{({\bf x}_t,{\bf y}_t):~t\in\mathbb{N}_0\}$ achieves a smaller value of the $\nrdf$, and if ${\bf x}^n$ is Gaussian and Markov, then the infimum in the $\nrdf$ can be restricted to test channel distributions which are Gaussian, of the form  ${\bf P}^{\gp}(dy_t|y_{t-1},x_t)$,  with  linear mean in $(x_t, y_{t-1})$ and conditional covariance which is non-random,~$t\in\mathbb{N}_0^n$. The above results are also derived in \cite[Theorem 5]{gorbunov-pinsker1974} for pointwise $\mse$ distortion constraint.}
\par {In view of the above results, the following hold.
\begin{problem}(A lower bound on Problem \ref{problem1})\label{problem2}{\ \\}
Consider the vector-valued Gaussian source model in \eqref{state_space_model}. Then, the finite-time Gaussian $\nrdf$ is characterized by the expression
\begin{align}
{R}^{\na}_{\text{GM},0,n}(D) \triangleq & \inf_{\substack{{\bf P}^{\gp}(dy_t|y_{t-1},x_t):~t\in\mathbb{N}_0^n\\\frac{1}{n+1}\mathbb{E}\left\{\sum_{t=0}^n||{\bf x}_t-{\bf y}_t||_2^2\right\}\leq{D}}}\frac{1}{n+1}I^{\overrightarrow{\bf Q}}({\bf x}^n;{\bf y}^n).\label{prob_stat:eq.2}
\end{align}
provided the infimum is achieved in the set.\\
The asymptotic limit of \eqref{prob_stat:eq.2} is defined as:
\begin{align}
{R}_{\text{GM}}^{\na}(D) &\triangleq \lim_{n\longrightarrow\infty}{R}^{\na}_{\text{GM},0,n}(D),\label{infinite:eq.2}
\end{align}
provided the infimum is achieved in the set and the limit exists and it is finite.
If the source model of \eqref{state_space_model} is stationary (or asymptotically stationary) then $R_{\text{GM}}^{\na}(D)=\widehat{R}_{\text{GM}}^{\na}(D)$ (see \cite[Theorem 4]{gorbunov-pinsker1974}), where $\widehat{R}_{\text{GM}}^{\na}(D)$ is defined as in \eqref{infinite_time_nrdf_upper_bound} but $\overrightarrow{\bf Q}^{\gp}(dy^\infty|x^\infty)$ denotes the sequence of conditional probability distributions $\{{\bf P}^{\gp}(dy_t|y_{t-1},x_t):~t\in\mathbb{N}_0\}$.
\end{problem}
}
\par{The next theorem, provides a series of inequalities that connect all previously discussed information measures in the context of Gaussian sources with asymptotic $\mse$ distortions.
\begin{theorem}(Inequalities)\label{theorem:bounds}{\ \\}
\noi For Gaussian sources with asymptotic $\mse$ distortion constraint, the following bounds hold.
\begin{align}
R(D)\leq{R}^{\na}(D)\leq\widehat{R}^{\na}(D)\leq{R}^{\op}_c(D)\leq{R}_{\zd}^{\op}(D).\label{bounds}
\end{align}  
where $R(D)$ denotes the classical $\rdf$ \cite{berger:1971}.
\end{theorem}
\begin{proof}
The bounds  are derived in \cite[eq. (11)]{derpich:2012} (the first two inequalities are also derived in  \cite{gorbunov-pinsker1973}). The last inequality follows by definition of the operational zero-delay Gaussian $\rdf$ and causal $\opta$ function.
\end{proof}
In the next remark, we state a bound on $R^{\na}_{\text{GM}}(D)$ for unstable Gauss-Markov sources and asymptotic $\mse$ distortion.
\begin{remark}(Bound on unstable $\mathbb{R}^p$-valued Gauss-Markov sources)\label{remark:1}{\ \\}
Consider the time-invariant vector-valued Gauss-Markov source of \eqref{state_space_model} where $A$ has eigenvalues with magnitude greater than one and the asymptotic $\mse$ distortion. {Then from  \cite{tatikonda-mitter2004}} 
\begin{align}
{R}_{\text{GM}}^{\na}(D)\geq\sum_{\mu_{A,i}>1}\log|\mu_{A,i}|.\label{unstable_eig}
\end{align}
\end{remark}
}
\subsection{Characterization of Problem \ref{problem2} via filtering and Markov realization}\label{subsec:specific_prob}

\par {The what follows, we leverage the Markov realization of the optimal test-channel that corresponds to Problem \ref{problem2}, \eqref{prob_stat:eq.2} {to provide} the complete characterization of Problem \ref{problem2}. We note that the following two results are derived in \cite{stavrou-charalambous-charalambous-loyka2018siam} but  we provide them herein for completeness.}
\par {The first result serves as an intermediate step towards the complete characterization of Problem \ref{problem2} and is a simple extension of the result derived for the scalar case in \cite[Lemma 1]{stavrou-charalambous-charalambous:2018lcss} hence we omit its proof.}
\begin{lemma}(Realization of $\{{\bf P}^*(dy_t|y^{t-1},x_t):~t\in\mathbb{N}_0^n\}$){\  \\}
\label{lemma:conditionally_gaussian}
Consider the class of test channels $\{{\bf P}^*(dy_t|y^{t-1},x_t):~t\in\mathbb{N}_0^n\}$. Then, the following statements hold.\\
(1) Any candidate of $\{{\bf P}^*(dy_t|y^{t-1},x_t):~t\in\mathbb{N}^n_0\}$ is realized by the recursion 
\begin{align}
{\bf y}_t=& H_t\left({\bf x}_t -\widehat{{\bf x}}_{t|t-1}\right)+ \widehat{{\bf x}}_{t|t-1}+ {\bf v}_t, \ \ t\in\mathbb{N}^n_0 \label{CG_1}
\end{align}
where $\widehat{{\bf x}}_{t|t-1}\triangleq\mathbb{E}\{{\bf x}_t|{\bf y}^{t-1}\}$,~$\{{\bf v}_t\in\mathbb{R}^p\sim{\cal N}(0; \Sigma_{{\bf v}_t}): ~t\in\mathbb{N}^n_0\}$ is an independent Gaussian process independent of $\{{\bf w}_t: ~t \in {\mathbb N}_0^{n-1}\}$ and ${\bf x}_0$, and  $\{H_t\in\mathbb{R}^{p\times{p}}:~~t\in\mathbb{N}^n_0\}$ are time-varying deterministic matrices. \\
Moreover, the innovations process $\{\tilde{\bf k}_t\in\mathbb{R}^p:~t\in\mathbb{N}^n_0\}$ of (\ref{CG_1}) is the orthogonal process defined by    
\begin{align}
\tilde{\bf k}_t&\triangleq {\bf y}_t-\mathbb{E}\left\{{\bf y}_t|{\bf y}^{t-1}\right\} \nonumber \\
&={\bf y}_t-\widehat{{\bf x}}_{t|t-1} =H_t \left({\bf x}_t -\widehat{{\bf x}}_{t|t-1} \right)+ {\bf v}_{t},  \label{inn_2} 
 \end{align}
where  $\tilde{\bf k}_t\sim{\cal N}(0;\Sigma_{\tilde{\bf k}_t})$,~$\Sigma_{\tilde{\bf k}_t}=H_t\Sigma_{t|t-1}H_t\T+\Sigma_{{\bf v}_t}$ and~$\Sigma_{t|t-1}\triangleq\mathbb{E}\left\{({\bf x}_t-\widehat{\bf x}_{t|t-1})({\bf x}_t-\widehat{\bf x}_{t|t-1})\T|{\bf y}^{t-1}\right\}$. \\
(2) Let $\widehat{\bf x}_{t|t}\triangleq\mathbb{E}\{{\bf x}_t|{\bf y}^{t}\}$ and $\Sigma_{t|t}\triangleq\mathbb{E}\left\{({\bf x}_t-\widehat{\bf x}_{t|t})({\bf x}_t-\widehat{\bf x}_{t|t})\T|{\bf y}^{t}\right\}$. Then, $\{\widehat{\bf x}_{t|t-1},~\Sigma_{t|t-1}:  t\in\mathbb{N}^n_0\}$ satisfy the following vector-valued {equations}:
\begin{subequations}
\label{kalman_filter}
\begin{align}
~&\widehat{\bf x}_{t|t-1}=A\widehat{\bf x}_{t-1|t-1},\label{kalman:1}\\
&\Sigma_{t|t-1}={A}\Sigma_{t-1|t-1}{A}\T+BB\T,\label{kalman:2}\\
~&\widehat{\bf x}_{t|t}=\widehat{\bf x}_{t|t-1}+{N}_t\tilde{\bf k}_t, \label{kalman:3} \\
&{N}_t=\Sigma_{t|t-1}H_t\T\Sigma_{\tilde{\bf k}_t}^{-1} ~\mbox{(Kalman ~Gain)},  \label{kalman:3a}  \\
&\Sigma_{t|t}=\Sigma_{t|t-1}-\Sigma_{t|t-1}H_t\T\Sigma_{\tilde{\bf k}_t}^{-1}H_t\Sigma_{t|t-1},\label{kalman:4}
\end{align}
\end{subequations}
where $\Sigma_{t|t}=\Sigma_{t|t}\T\succeq{0}$ and $\Sigma_{t|t-1}=\Sigma_{t|t-1}\T\succeq{0}$.\\
(3) $R^{\na}_{\text{GM},0,n}(D)$ is given by
\begin{subequations}
\label{initial:optimization}
\begin{align*}
R^{\na}_{\text{GM},0,n}(D)&= \inf_{H_t\succeq{0},~\Sigma_{{\bf v}_t}\succeq{0},~t\in\mathbb{N}_0^n} \frac{1}{2}\frac{1}{n+1}\sum_{t=0}^n\left[\log\frac{|\Sigma_{t|t-1}|}{|\Sigma_{t|t}|}\right]^{+},\\
&\qquad\Sigma_{t|t-1}~\text{is given by \eqref{kalman:2}}\\
&\qquad\Sigma_{t|t}~\text{is given by \eqref{kalman:4}}\\
\frac{1}{n+1}\sum_{t=0}^n&\trace\left((I-H_t)\Sigma_{t|t-1}(I-H_t)\T+\Sigma_{{\bf v}_t}\right)\leq{D}
\end{align*} 
\end{subequations}
for some $D \in  [0, \infty]$ and $[x]^{+}\triangleq \max \{0,x\}$.
\end{lemma}
\par {The next theorem uses Lemma \ref{lemma:conditionally_gaussian} to identify the missing parameters in the realization of \cite[Theorem 5]{gorbunov-pinsker1974} and, therefore, to provide the complete characterization of $R_{\text{GM},0,n}^{\na}(D)$.}
{
\begin{theorem}(Characterization of Gaussian $\nrdf$)\label{theorem:alternative_optimization}{\ \\}
Consider Problem \ref{problem2}, \eqref{prob_stat:eq.2}. Then, the following holds.\\
(1) The ``test channel'' ${\bf P}(dy_t|y^{t-1},x_t)={\bf P}^{\gp}(dy_t|y_{t-1},x_t)$ and is realized by  
\begin{align}
{\bf y}_t={H}_t{\bf x}_t+(I-{H}_t)A{\bf y}_{t-1}+{\bf v}_t,~{\bf y}_{-1}=\bar{y},~t\in\mathbb{N}_0^n,\label{optimal_realization}
\end{align}
where
\begin{subequations}
\label{scalings}
\begin{align}
&H_t\triangleq{I}-\Sigma_{t|t}\Sigma^{-1}_{t|t-1}\succeq 0,~~\Sigma_{t|t}\succeq{0},~\Sigma_{t|t-1}\succeq{0},\label{scalings:1}\\
&\Sigma_{{\bf v}_t}\triangleq\Sigma_{t|t}H_t\T\succeq 0,\label{scalings:2}\\
&\Sigma_{t|t-1}~~\text{satisfies~ (\ref{kalman:2})},~~\Sigma_{0|-1}=\Sigma_{{\bf x}_0}.\label{scalings:3}
\end{align}
\end{subequations}
(2) Moreover, the above realization yields in Lemma \ref{lemma:conditionally_gaussian}
\begin{align}
&\widehat{\bf x}_{t|t}={\bf y}_t,~\widehat{\bf x}_{t|t-1}=A {\bf y}_{t-1}. \label{sim} 
\end{align}
(3) The characterization of $R_{\text{GM},0,n}^{\na}(D)$ is
\begin{subequations}
\label{alternative:optimization}
\begin{align}
R^{\na}_{\text{GM},0,n}&(D)=  \inf_{\Sigma_{t|t}\succeq{0}}\frac{1}{2}\frac{1}{n+1}\sum_{t=0}^n  \left[\log\frac{|\Sigma_{t|t-1}|}{|\Sigma_{t|t}|}\right]^{+},\label{scalar:eq.3a}\\
\text{s.t.}~
&0\preceq\Sigma_{t|t}\preceq\Sigma_{t|t-1},~t\in\mathbb{N}_0^n\label{vector_lmi}\\
&\frac{1}{n+1}\sum_{t=0}^n\trace(\Sigma_{t|t})\leq{D}\label{scalar:eq.2b}
\end{align}
\end{subequations}
for some $D \in [0, \infty]$.
\end{theorem}
\begin{proof} 
The proof is derived in \cite{stavrou-charalambous-charalambous-loyka2018siam}.
\end{proof}
}
\par {Next, we give sufficient conditions for existence of solution to Theorem \ref{theorem:alternative_optimization}, (3).
\begin{remark}(Existence of solution of \eqref{alternative:optimization})\label{remark:2}{\ \\}
A sufficient condition for existence of solution with finite value in \eqref{alternative:optimization} is to consider the strict linear matrix inequality ($\lmi$) constraint in \eqref{vector_lmi}, i.e., $0\prec\Sigma_{t|t}\preceq\Sigma_{t|t-1},~\forall{t}\in\mathbb{N}_0^n$,  because otherwise the value of $\nrdf$ takes the value of $+\infty$. Then, by construction, the minimization problem of \eqref{alternative:optimization} is strictly feasible, i.e., there always exists an optimal solution with finite value. The strict $\lmi$ further means that $D>0$ (non-zero distortion) and also $\Sigma_{t|t-1}\succ{0}$.  Then, from \eqref{kalman:2} the following conditions on matrices $A$ and $B$ are sufficient for existence of a finite solution: 
\begin{align}
\mbox{either $A$ is full rank or $B$ is square and full rank in \eqref{state_space_model}}. \label{cond_1}
\end{align}
\end{remark}
}
\section{Asymptotic Feedback Realization Scheme via Kalman Filtering for Problem \ref{problem2} }\label{sec:new_realization_scheme}

\par {In this section, we leverage results from \S \ref{subsec:specific_prob} to show that the asymptotic limit of \eqref{alternative:optimization} exists and it is finite. Then, we propose a new alternative realization scheme that makes use of joint diagonalization matrices, reverse-waterfilling design parameters by means of an \textit{innovations encoder}, an \textit{additive Gaussian channel}, and a \textit{decoder} which includes a Kalman filter. Recall that our feedback realization scheme is fundamentally different compared to the approach considered in \cite{tanaka:2017} because it builts upon the realization scheme of \cite[Theorem 5]{gorbunov-pinsker1974}  whereas the one in \cite{tanaka:2017} makes use of the so called ``sensor-estimation separation principle''.} 

\par {In the first result of this section, we provide sufficient conditions to show existence of solution with finite value to the asymptotic limit of \eqref{alternative:optimization} and then we give the asymptotic characterization of Theorem \ref{theorem:alternative_optimization}.  
\begin{theorem}(Existence of solution to the asymptotic characterization of \eqref{alternative:optimization})\label{theorem:asymptotic_limit}{\ \\}
Suppose condition (\ref{cond_1}) holds and the optimal test channel distribution ${\bf P}^{\gp}(dy_t|y_{t-1},x_t)$ is restricted to be time invariant and there is a unique invariant distribution of the transition probability ${\bf P}(dy_t|y_{t-1})$.  Then, the following statements hold.\\
(1) The limit
\begin{align}
R^{\na}_{\text{GM}}(D)=\lim_{n\longrightarrow\infty}R_{\text{GM},0,n}^{\na}(D)<\infty,\label{existence_of_limit}
\end{align}
i.e., exists and it is finite, if the solution  of the limiting problem   $R^{\na}_{\text{GM}}(D)$ for $D \in (0, \infty]$, is finite,    given by
\begin{align}
R^{\na}_{\text{GM}}&(D)=  \min_{\substack{0\prec\Pi\preceq\Lambda\\ \trace(\Delta)\leq{D}}} \frac{1}{2}\log\frac{|\Lambda|}{|\Pi|},\label{alternative:optimization:inf}
\end{align}
where $(\Lambda,~\Pi)$  are the corresponding time-invariant values of $\Sigma_{t|t-1}$ and $\Sigma_{t|t}$, respectively.\\
(2) The asymptotic limit of ${\bf P}^{\gp}(dy_t|y_{t-1},x_t)$ is realized by  
\begin{align}
{\bf y}_t={H}{\bf x}_t+(I-{H})A{\bf y}_{t-1}+{\bf v}_t,\label{optimal_realization_asymptotic}
\end{align}
where ${\bf v}_t\sim{\cal N}(0;\Sigma_{\bf v})$,
\begin{subequations}
\label{scalings_inf}
\begin{align}
&H\triangleq{I}-\Pi\Lambda^{-1}\succeq 0,~~\Pi\succeq {0},~\Lambda\succeq {0},\label{scalings:1:inf}\\
&\Sigma_{{\bf v}}\triangleq\Pi{H}\T\succeq 0,\label{scalings:2:inf}\\
&\Lambda={A}\Pi{A}\T+BB\T,\label{scalings:3:inf}
\end{align}
and ($H$,~$\Sigma_{\bf v}$) are the corresponding time-invariant values of $H_t$ and $\Sigma_{{\bf v}_t}$, respectively.
\end{subequations}
\end{theorem}
\begin{proof} 
(1) Observe that the sequence $\left\{R_{\text{GM},0,n}^{\na}(D): n\in\mathbb{N}_0^n \right\}$ is sub-additive (see \cite[Lemma 1]{gorbunov-pinsker1973}). Hence, the limit in \eqref{existence_of_limit} always exists (although it can be infinite). However, since we assumed the optimal test channel distribution is time-invariant and there is a unique invariant distribution, we ensure that the limit is finite. The last part follows because we assume ${\bf P}^{\gp}(dy_t|y_{t-1},x_t)$ is time-invariant. (2) follows from (1) and Theorem \ref{theorem:alternative_optimization}. 
This completes the proof.
\end{proof}
The minimization problem of Theorem \ref{theorem:asymptotic_limit}, \eqref{alternative:optimization:inf} can be solved by employing, for instance, Karush-Kuhn-Tucker ($\kkt$) conditions \cite{stavrou-charalambous-charalambous-loyka2018siam} or semidefinite programming algorithm \cite{tanaka:2017}. We wish to remark that sufficient conditions which do not require the test channel in Theorem \ref{theorem:asymptotic_limit} to be time-invariant, can be identified upon solving the $\kkt$ conditions of the finite-time Gaussian $\nrdf$ of Theorem \ref{theorem:alternative_optimization} (through the solutions of the Riccati equations). This method is employed in \cite{stavrou-charalambous-charalambous-loyka2018siam}. In the next theorem, we evaluate numerically the optimization problem
\eqref{alternative:optimization:inf} by providing two equivalent semidefinite programming representations of $R_{\text{GM}}^{\na}(D)$. The first is similar to the one derived in \cite[equation (27)]{tanaka:2017} whereas the second is new. The utility of each of these semidefinite representations will be perceived in the sequel.
\begin{lemma}(Optimal solution of $R_{\text{GM}}^{\na}(D)$)\label{theorem:solution:semidefinite}{\ \\}
Suppose the conditions of Theorem \ref{theorem:asymptotic_limit} are satisfied. Then, the following statements hold.
\item[{\bf (1)}] Suppose matrix $B$ is full rank. Introduce the variable $Q_1\triangleq{\Pi}^{-1}-A{\T}(BB{\T})^{-1}A$, where $\Pi\succ{0}$. Then, for $D>0$, $R_{\text{GM}}^{\na}(D)$ is semidefinite representable as follows: 
\begin{subequations}\label{sdp:ver1}
\begin{align}
R_{\text{GM}}^{\na}(D)&= \min_{Q_1\succ 0} - \frac{1}{2} \log|Q_1| + \frac{1}{2} \log |BB{\T}|. \label{eq:semidefinite_representation1} \\
\text{s.t. } &  \;\; 0\prec\Pi\preceq \Lambda\\
& \trace(\Pi) \leq D \\
& \left[ \begin{array}{cc}
\Pi-Q_1 & \Pi{A}{\T} \\
A\Pi & \Lambda \end{array}\right]\succeq 0 \label{lmi:1}
\end{align}
\end{subequations}
\item[{\bf(2)}] Suppose matrix $A$ is full rank. Introduce the decision variable $Q_2\triangleq{I}+B{\T}(A\T)^{-1}\Pi^{-1}A^{-1}B$, where $\Pi\succ{0}$. Then, for $D>0$, $R_{\text{GM}}^{\na}(D)$ is semidefinite representable as follows:
\begin{subequations}\label{sdp:ver2}
\begin{align}
R_{\text{GM}}^{\na}(D)&= \min_{Q_2\succ 0} - \frac{1}{2} \log|Q_2| + \log\abs(|A|). \label{eq:semidefinite_representation2} \\
\text{s.t. } &  \;\; 0\prec\Pi\preceq\Lambda\\
& \trace(\Pi) \leq D \\
& \left[ \begin{array}{cc}
I-Q_2 & B{\T} \\
B & \Lambda \end{array}\right]\succeq 0 \label{lmi:2}
\end{align}
\end{subequations}
\end{lemma}}
\begin{proof}
{The derivation of the semidefinite representation of {\bf (1)} follows similar to \cite[equation (27)]{tanaka:2017}, hence we omit it. The derivation of {\bf (2)} is given in Appendix~\ref{proof:theorem:solution:semidefinite}}.
\end{proof}
The two semidefinite representations of $R^{\na}(D)$ in Lemma \ref{theorem:solution:semidefinite} gives the flexibility of working with different assumptions on matrices $A$ and $B$. The first semidefinite representation of $R_{\text{GM}}^{\na}(D)$ is suitable to evaluate a  $\mathbb{R}^p$-valued Gauss-Markov source of any dimension $p$. However, it excludes the possibility of evaluating $\mathbb{R}^p$-valued Gauss-Markov sources where matrix $B$ is singular. In this case, the second representation of $R_{\text{GM}}^{\na}(D)$ is the suitable one. However, even for this case, the restriction is that matrix $A$ has to be full rank. {These conditions on matrices $A$ and $B$ essentially guarantee existence of a solution with finite value (by construction) for the asymptotic limit (cf. Remark \ref{remark:2}).} 
\par{The next theorem is the main result of this section. We provide an equivalent alternative realization to \eqref{optimal_realization_asymptotic}. This realization reveals the reverse-waterfilling solution (in dimension) that characterizes the optimization problem \eqref{alternative:optimization:inf}. To the best of our knowledge this approach is new and has never been documented elsewhere. 
\begin{theorem}(Equivalent realization scheme to \eqref{optimal_realization_asymptotic})\label{theorem:alter:real:waterfilling}{\ \\}
An equivalent realization scheme to \eqref{optimal_realization_asymptotic} is the following:
\begin{align}
{\bf y}_t=E^{-1}\tilde{H}E{\bf x}_t+(I-E^{-1}\tilde{H}E)A{\bf y}_{t-1}+E^{-1}\Theta{\bf v}_t,\label{optimal:realization:diagonal}
\end{align} 
where ${\bf v}_t\sim{\cal N}(0;I),~I\in\mathbb{R}^{p\times{p}}$, $E\in\mathbb{R}^{p\times{p}}$ is a non-singular matrix that simultaneously diagonalizes $\Pi\succ{0},~\Lambda\succ{0}$, such that
\begin{align}
\tilde{\Lambda}\triangleq{E}\Lambda{E}\T\equiv\diag(\mu_{\Lambda,i}),~~\tilde{\Pi}\triangleq{E}\Pi{E}\T\equiv\diag(\mu_{\Pi,i}),\label{diagonalized_lambdas_deltas}
\end{align}
with $\mu_{\Lambda,1}\geq\mu_{\Lambda,2}\geq\ldots\geq\mu_{\Lambda,p}$ and $\mu_{\Pi,1}\geq\mu_{\Pi,2}\geq\ldots\geq\mu_{\Pi,p}$; the reverse-waterfilling design matrix $\tilde{H}\in\mathbb{R}^{p\times{p}}$ is defined as:
\begin{subequations}\label{scalings:no_waterfilling}
\begin{align}
&\tilde{H}=I-\tilde{\Pi}\tilde{\Lambda}^{-1}\equiv{\Theta}{\Phi},\label{design_parameters_no_waterfilling2}\\
&\Theta={\tilde{\Sigma}_{\bf v}}^{\frac{1}{2}},~\tilde{\Sigma}_{\bf v}=\tilde{\Pi}\tilde{H},~\Theta\in\mathbb{R}^{p\times{p}},~\label{scaling_no_waterfilling1}\\
&\Phi=({\tilde{H}\tilde{\Pi}^{-1}})^{\frac{1}{2}},~\Phi\in\mathbb{R}^{p\times{p}}.~\label{scaling_no_waterfilling2}
\end{align} 
\end{subequations}
\paragraph{Full-rank $\tilde{H}$} If $\tilde{H}\succ{0}$, i.e., $\tilde{H}\in\mathbb{R}^{p\times{p}}$, where $r\triangleq\rank(\tilde{H})=p$, then, no reverse-waterfilling occurs (in dimension) and $\Theta\succ{0}$, $\Phi\succ{0}$.
\paragraph{Rank-deficient $\tilde{H}$} If $\tilde{H}\succeq{0}$, i.e., $\tilde{H}\in\mathbb{R}^{p\times{p}}$,~where $r<p$, then, the reverse-waterfilling kicks in and $\Theta\succeq{0}$, $\Phi\succeq{0}$.
\end{theorem}
\begin{proof}
See Appendix \ref{proof:theorem:alter:real:waterfilling}.
\end{proof}
}

{The realization scheme of Theorem \ref{theorem:alter:real:waterfilling}, \eqref{optimal:realization:diagonal} is illustrated in Fig. \ref{fig:noisy_communication_system}. Next, we briefly discuss the basic features of this scheme.}
\begin{figure*}[htp]
\centering
\includegraphics[width=\textwidth]{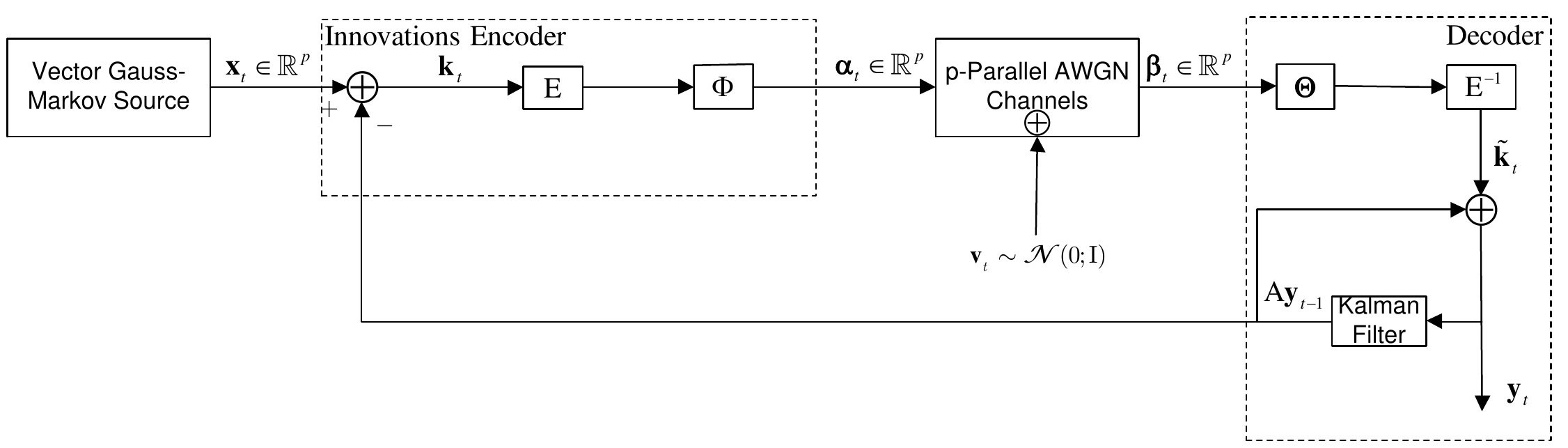}
\caption{Asymptotic feedback realization scheme of the optimal ``test-channel'' ${\bf P}^{\gp}(y_t|y_{t-1},x_t)$.}\label{fig:noisy_communication_system}
\end{figure*}
{\paragraph*{Realization scheme of Fig. \ref{fig:noisy_communication_system}} The {\it encoder} does not directly transmit the $\mathbb{R}^{p}-$valued Gauss-Markov process ${\bf x}_t$. Rather, it conveys the deviation from the linear estimate $\widehat{x}_{t|t-1}=A{\bf y}_{t-1}$ (see Theorem \ref{theorem:alternative_optimization}, \eqref{sim}) of ${\bf x}_t$ denoted by $\{{\bf k}_t\in\mathbb{R}^p:~{t\in\mathbb{N}_0}\}$, where ${\bf k}_t\triangleq{\bf x}_t-\widehat{\bf x}_{t|t-1}\sim{\cal N}(0;\Lambda)$. This method is known in least-squares estimation theory as {\it innovations approach} \cite{kailath1968} hence the encoder is an {\it innovations encoder}. The ``error'' process ${\bf k}_t$ has correlated temporal and spatial Gaussian components. However, by introducing the non-singular matrix $E$ we create $p$ independent spatial Gaussian components which are then scaled by the reverse-waterfilling design matrix $\Phi$ to create an independent (in dimension) Gaussian process $\{{\bm \alpha}_t\in\mathbb{R}^p:~t\in\mathbb{N}_0\}$. The resulting parallel Gaussian process is then conveyed through $p$-parallel $\awgn$ channels. This is compactly written as follows: 

\begin{align}
{\bm \beta}_t\triangleq{\bm \alpha}_t+{\bf v}_t,~~{\bf v}_t\sim{\cal N}(0;I),~t\in\mathbb{N}_0. \label{awgn:eq.1}
\end{align}
The output process $\{{\bm \beta}_t:~t\in\mathbb{N}_0\}$ is then scaled with another reverse-waterfilling design matrix $\Theta$. Afterwards, the invertible linear operator $E^{-1}$ is introduced to transform the independent spatial Gaussian components to correlated spatial components. At the {\it decoder} the resulting process is the innovations process $\{{\bf \tilde{k}}_t:~{t\in\mathbb{N}_0}\}$ on which we add $\widehat{x}_{t|t-1}=A{\bf y}_{t-1}$, to obtain the {\it estimate} $\widehat{x}_{t|t}={\bf y}_t$.
\paragraph*{Reverse-waterfilling} If Theorem \ref{theorem:alter:real:waterfilling}, a) occurs, then, $\tilde{H}\succ{0}$. This means that all dimensions are active, therefore, no reverse-waterfilling occurs (full rank case of $\Theta,~\Phi$ with $r=p$). If Theorem \ref{theorem:alter:real:waterfilling}, b) occurs, then, $\tilde{H}\succeq{0}$. This in turn means that the reverse-waterfilling (in dimension) kicks in and as a result certain dimensions are inactive, i.e., they convey zero-rate (rank-deficient case of $\Theta,~\Phi$ with $r<p$, where $p-r$ are the dimensions with zero information rate).}
\par {The following corollary demonstrates certain data processing equalities of the feedback realization scheme of Fig. \ref{fig:noisy_communication_system}. These are slightly different depending on whether the reverse-waterfilling kicks in.
\begin{corollary}(Data processing equalities)\label{corollary:dpe}{\ \\}
Consider the realization scheme of Fig. \ref{fig:noisy_communication_system}. Then the following data processing equalities hold.\\
{(1)} If $\Phi\succ{0},~\Theta\succ{0}$, then, 
\begin{align}
I({\bf x}_t;{\bf y}_t|{\bf y}_{t-1})\stackrel{(a)}=I({\bf k}_t;\tilde{\bf k}_t|{\bf y}_{t-1})\stackrel{(b)}=I({\bf k}_t;\tilde{\bf k}_t)\stackrel{(c)}=I({\bm \alpha}_t;{\bm \beta}_t),\label{dpe:eq1}
\end{align}
where ${\bf x}_t\in\mathbb{R}^p, {\bf y}_t\in\mathbb{R}^p, {\bf k}_t\in\mathbb{R}^p,\tilde{\bf k}_t\in\mathbb{R}^p,~{\bm \alpha}_t\in\mathbb{R}^p,~{\bm \beta}_t\in\mathbb{R}^p$.\\
{(2)} If $\Phi\succeq{0},~\Theta\succeq{0}$, then, \eqref{dpe:eq1} holds if ${\bf x}_t\in\mathbb{R}^r, {\bf y}_t\in\mathbb{R}^r, {\bf k}_t\in\mathbb{R}^r,\tilde{\bf k}_t\in\mathbb{R}^r,~{\bm \alpha}_t\in\mathbb{R}^r,~{\bm \beta}_t\in\mathbb{R}^r$, i.e., when the $p-r$ (inactive) dimensions are excluded from the realization scheme.
\end{corollary}
\begin{proof}
(1) This holds in the absence of the reverse-waterfilling in dimension. Equality $(a)$ holds because ${\bf k}_t={\bf x}_t-A{\bf y}_{t-1}$ and $\tilde{k}_t={\bf y}_t-A{\bf y}_{t-1}$ where $A{\bf y}_{t-1}$ is measurable with respect to the $\sigma-$algebra (information) generated by ${\bf y}_{t-1}$; equality $(b)$ holds because the error process ${\bf k}_t$ and the innovations process $\tilde{\bf k}_t$ are orthogonal to ${\bf y}_{t-1}$; equality $(c)$ holds because ($E,~\Phi, \Theta$) are invertible matrices. (2) This holds if the reverse-waterfilling in dimension kicks in. Equalities $(a), (b)$ hold similarly to (1). Equality $(c)$ holds if $\Phi$ and $\Theta$ are full rank. This is established if we ``remove'' from the system the $p-r$ (inactive) dimensions that convey zero rate (because of the reverse-waterfilling). This completes the proof.
\end{proof}
} 
\par In the next remark, we show how to recover from our realization scheme the closed form expression of $R_{\text{GM}}^{\na}(D)$ that corresponds to a time-invariant scalar-valued Gaussian $\ar(1)$ source. This result appeared in many papers, see for instance, \cite{gorbunov-pinsker1974,tatikonda:2004,derpich:2012,stavrou-charalambous-charalambous:2018lcss} but we include it here for completeness.
\begin{remark}(Scalar-valued Gaussian $\ar(1)$ process)\label{remark:3}{\ \\}
Consider the scalar case of \eqref{state_space_model} with $A=\alpha\in\mathbb{R}$  and $B{\bf w}_t\sim{\cal N}(0;\sigma^2_{\bf w})$. This special case applied in Theorem \ref{theorem:alter:real:waterfilling} corresponds to $E=1$, $\Pi=\tilde{\Pi}=D$ and $\Lambda=\alpha^2D+\sigma^2_{\bf w}$. This means that \eqref{alternative:optimization:inf} becomes
\begin{equation}
{R}_{\text{GM}}^{\na}(D)=\frac{1}{2}\log\left(\frac{\alpha^2{D}+\sigma^2_{\bf w}}{D}\right)=\frac{1}{2}\log\left(\alpha^2+\frac{\sigma^2_{\bf w}}{D}\right).\label{optimal_solution:scalar:eq.1}
\end{equation}
For a stationary stable source, i.e., source where $|a|<1$, it was shown in \cite{gorbunov-pinsker1974,derpich:2012} that \eqref{optimal_solution:scalar:eq.1} has $D_{\max}=\frac{\sigma^2_{\bf w}}{1-\alpha^2}$. 
\end{remark}

%
%
%
%

\section{Upper Bounds on {Problem \ref{problem1}} via {Problem \ref{problem2}}}\label{sec:main_results}

\par In this section, we use the asymptotic feedback realization scheme of Fig.~\ref{fig:noisy_communication_system} {that corresponds to the optimal solution of Problem \ref{problem2}}, to construct a coding scheme with achievable upper bounds to {Problem \ref{problem1}}. The coding scheme employs joint entropy coded dithered quantization ($\ecdq$) and lattice codes. Further, we explain how to extend this result to time-invariant {vector-valued} Gauss-Markov processes of any finite order. The obtained theoretical upper bound is compared to existing bounds in the literature.

\subsection{Scalar Quantization}
\label{sec:upper_bounds_zero_delay_codes}

\par To achieve our goal, we invoke a universal quantization scheme based on a subtractive dither with uniform scalar quantization ($\sdusq$) \cite{zamir:2014} on the feedback realization scheme illustrated in Fig.~\ref{fig:noisy_communication_system}. The $\sdusq$ approach is common in the literature. However, it has never been documented elsewhere for the proposed realization setup. Some recent works where the use of $\sdusq$ is demonstrated under various setups, are found in \cite{derpich:2012,yuksel:2014,tanaka:2016,rabi:2016}. 

\par Before we proceed, we state the definition of a uniform scalar quantizer with subtractive dither \cite{zamir:2014}.
\par A scalar quantizer function is defined as 
\begin{align}
Q_\Delta({\bf x})=j\Delta~~\mbox{for}~~j\Delta-\frac{\Delta}{2}\leq{{\bf x}}\leq{j}\Delta+\frac{\Delta}{2},~j\in\mathbb{Z},\nonumber
\end{align}
where $\Delta>0$ is the quantization step which is freely designed by the designer. 
A scalar universal uniform quantizer with subtractive dither $Q^{\sd}_\Delta({\bf x})$ is defined as 
\begin{align}
Q^{\sd}_\Delta({\bf x})=Q_\Delta({\bf x}+{\bf q})-{\bf q}, \label{def:scalar_uniform_dithered_new}
\end{align}
where ${\bf q}$ is the realization of a uniformly distributed random variable ${\bf R}$ over the interval $[-\frac{\Delta}{2},\frac{\Delta}{2}]$. 
\par The execution of $Q^{\sd}_\Delta(\cdot)$ requires a common randomness both at the encoder's and the decoder's ends. In practice, the dither ${\bf r}$ acts as a synchronized pseudo-random noise generator that can be used at both encoder and decoder's end.

\subsubsection*{Componentwise uniform scalar quantization}\label{subsection:componentwise_quantization}  

\par Next, we use the asymptotic feedback realization scheme illustrated in Fig. \ref{fig:noisy_communication_system} to design an efficient \texttt{\{encoder/quantizer,decoder\}} pair. This procedure is described next.
\par We select the quantizer step size $\Delta$ so that the covariance of the resulting quantization error meets {the covariance of the Gaussian noise ${\bf v}_t$, i.e.,  $\Sigma_{\bf v}=I$}. {Recall that the encoder in Fig. \ref{fig:noisy_communication_system} is an innovations quantizer in the sense that it does not quantize the observed state ${\bf x}_t$ directly. Instead, it quantizes the error process ${\bf k}_t$}. In itself, innovation quantizer is known to be optimal in various setups, see for instance \cite{yuksel:2014,rabi:2016}. However, the novelty here lies in being able to
bound the performance of this particular realization scheme. {Note that in what follows, we consider $\mathbb{R}^r-$valued processes to present a coding scheme that takes into account the effect of the spatial reverse-waterfilling. This means that if $r=p$, then, there is no reverse-waterfilling (i.e, $\Phi\succ{0},~\Theta\succ{0}$) whereas when $r<p$ it means that the reverse-waterfilling kicks in (i.e., $\Phi\succeq{0},~\Theta\succeq{0}$) (see \S\ref{sec:new_realization_scheme}).}   
\par  We consider the zero-delay source coding setup illustrated in Fig.~\ref{fig:noisy_communication_system} with the $r-$parallel $\awgn$ channel replaced by $r$ independently operating $\sdusq$. This change is illustrated in Fig.~\ref{fig:replacement}. Note that, all {the designed matrices} adopted in Fig. \ref{fig:noisy_communication_system} still hold when the aforementioned change is applied.

\begin{figure}
    \centering
    \begin{subfigure}[b]{0.52\columnwidth}
        \includegraphics[width=\textwidth]{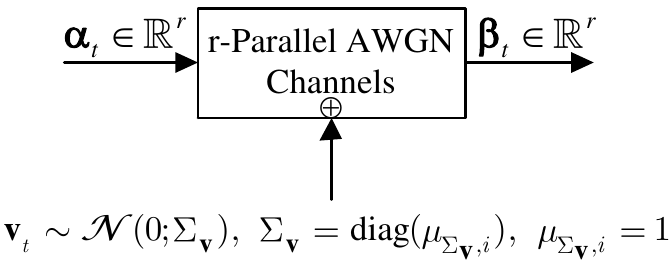}         
        \caption{$r-$parallel $\awgn$ channels.}
        \label{fig:awgn}
    \end{subfigure}
    \centering
        \begin{subfigure}[b]{\columnwidth}
        \includegraphics[width=\textwidth]{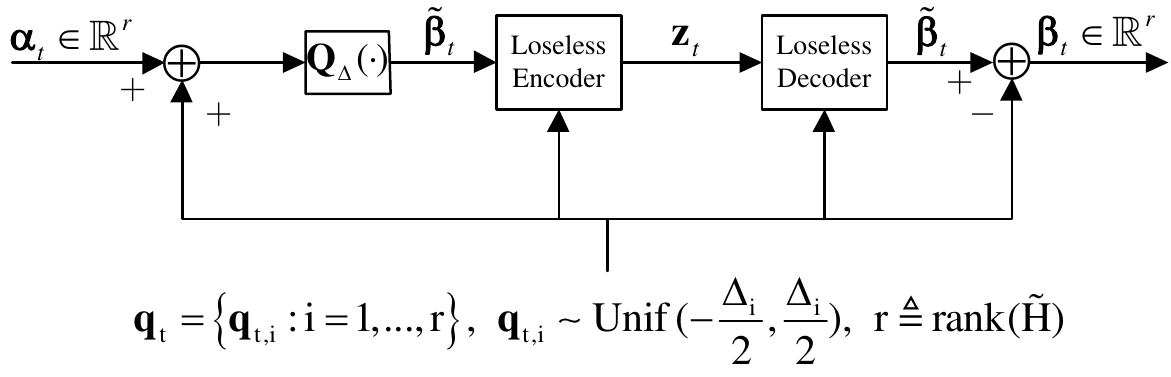}
        \caption{Realization over $r$ independently operating $\sdusq$.}
        \label{fig:sdusq}
    \end{subfigure}
\caption{Componentwise uniform scalar quantization by replacing $r-$ parallel $\awgn$ channels with $r$ independently operating $\sdusq$.}\label{fig:replacement}
\end{figure}
\par {Toward this end, for each time step $t$, the input to the quantizer, is the  scaled estimation error process $\{{\bm \alpha}_t\in\mathbb{R}^r:~t\in\mathbb{N}_0\}$ with $r$-independent spacial components which instead to the $r-$parallel $\awgn$ channels is conveyed (following Fig. \ref{fig:replacement}) through $r$ independently operating $\sdusq$. The goal it to design the covariance matrix $\Sigma_{\bf v}$ of the $\awgn$ corresponding to the $r$-parallel $\awgn$ channels in such a way, that for each $t$, each diagonal entry of $\mu_{\Sigma_{\bf v},i}=1, \forall{i}=1,\ldots,r$, i.e., $\Sigma_{\bf v}\triangleq{I}$ to correspond to a quantization step size $\Delta_i,~ i=1,\ldots,r$, such that} 
\begin{align}
{\mu_{\Sigma_{\bf v},i}}=\frac{\Delta^2_i}{12},~i=1,\ldots,r,~~D=\sum_{i=1}^r\frac{\Delta^2_i}{12},~{r\leq{p}}.\label{quantization_step}
\end{align}
This results into creating a multi-input multi-output ($\mimo$) transmission of parallel and independent $\sdusq$. We apply $\sdusq$ to each component of ${\bm \alpha}_t$, i.e.,
\begin{align}
{\bm \beta}_{t,i}=Q^{\sd}_{\Delta_i}({\bm \alpha}_{t,i}),~i=1,\ldots,{r},\label{output_quantizer}
\end{align}
and we let ${\bf q}_t$ be the $\mathbb{R}^{{r}}-$valued random process of dither signals whose individual components  $\{{\bf q}_{t,1},\ldots,{\bf q}_{t,{r}}\}$ are mutually independent and uniformly distributed $\rvs$ ${\bf q}_{t,i}\sim{\unif}(-\frac{\Delta_i}{2},\frac{\Delta_i}{2})$ independent of the corresponding source input components ${\bm{\alpha}}_{t,i},~\forall{t,i}$. The output of the quantizer is given by
\begin{align}
\tilde{\bm \beta}_{t,i}=Q_{\Delta_i}({\bm \alpha}_{t,i}+{\bf q}_{t,i}),~i=1,\ldots,{r}.\label{output_quantizer1} 
\end{align}
Note that $\tilde{\bm \beta}_{t}=\{\tilde{\bm \beta}_{t,1},\ldots,\tilde{\bm \beta}_{t,{r}}\}$ can take a countable number of possible values. In addition, by construction (see Fig.~\ref{fig:noisy_communication_system}), the sequences $\{{\bm \alpha}_t:~t\in\mathbb{N}_0\}$ and $\{\tilde{\bm \beta}_{t}:~t\in\mathbb{N}_0\}$ are not Gaussian any more since by applying the change illustrated in Fig.~\ref{fig:replacement}, $\{{\bm \alpha}_t:~t\in\mathbb{N}_0\}$ and $\{\tilde{\bm \beta}_{t}:~t\in\mathbb{N}_0\}$ contain samples of the uniformly distributed process $\{{\bm q}_t:~t\in\mathbb{N}_0\}$. As a result, the Kalman filter in Fig.~\ref{fig:noisy_communication_system} is no longer the least mean square estimator since the obtained quantized signals are no longer Gaussian.

\par For completeness, we illustrate in Fig.~\ref{fig:equivalent_additive_uniform_noise_channel} the relation between a $\sdusq$ and a uniform scalar additive noise channel \cite{zamir:2014}. This connection is further discussed in Remark \ref{remark:connection}.
\begin{remark}($\sdusq$ modeled as a uniform scalar additive noise channel)\label{remark:connection}{\ \\}
It is well known (see, e.g., \cite{zamir:2014}) that the $\sdusq$ considered in the realization of Fig.\ref{fig:noisy_communication_system} can be modeled as a uniform scalar noise channel. In both models, the source ${\bf x}_t$ and the output ${\bf y}_t$ possess the same statistics. The only difference in the realization scheme of Fig.~\ref{fig:noisy_communication_system} when the channels are $\awgn$ channels and when the channels are uniform additive noise channels as in Fig.~\ref{fig:equivalent_additive_uniform_noise_channel} lies on the statistics of the noise  process $\{{\bf v}_t:~t\in\mathbb{N}_0\}$ and $\{{\bm \xi}_t:~t\in\mathbb{N}_0\}$ respectively. In the latter case, the quantization noise is given by
\begin{align}
{\bm \xi}_t={\bm \beta}_t-{\bm \alpha}_t.\label{quantization_noise}
\end{align}
Note that ${\bm \xi}_t$ is orthogonal to ${\bf k}_{t^\prime}$ for any $t^\prime\leq{t}$, where $t^\prime=0,1,\ldots,t$.
\end{remark}
\paragraph*{Entropy coding} In what follows, we apply joint entropy coding across the vector dimension ${r}$ and memoryless coding across the time, that is, at each time step $t$ the output of the quantizer $\tilde{\bm \beta}_{t}$ is conditioned to the dither to generate a codeword ${\bf z}_t$. The decoder reproduces  ${\bm \beta}_{t}$ by subtracting the dithered signal ${{\bf q}_t}$ from $\tilde{\bm \beta}_{t}$. Specifically, at every time step $t$, we require that a message $\tilde{\bm \beta}_t$ is mapped into a codeword ${\bf z}_t\in\{0,1\}^{{\bf l}_t}$ designed using Shannon codes \cite[Chapter 5.4]{cover-thomas2006}. 
Recall that for a $\rv$ ${\bf x}$, the codes constructed based on Shannon coding scheme give an instantaneous (prefix-free) code with expected code length that satisfies the following bound:
\begin{align}
\mathbb{H}({\bf x})\leq\mathbb{E}({\bf l})\leq\mathbb{H}({\bf x})+1.\label{intantaneous_code_bounds}
\end{align}
If ${\bf x}\in\mathbb{R}^{{r}}$ then, the normalized version of \eqref{intantaneous_code_bounds} gives
\begin{align}
\frac{\mathbb{H}({\bf x})}{{r}}\leq\frac{\mathbb{E}({\bf l})}{{r}}\leq\frac{\mathbb{H}({\bf x})}{{r}}+\frac{1}{{r}}.\label{intantaneous_code_bounds:eq.1}
\end{align}
\begin{figure}[htp]
\centering
\includegraphics[width=0.5\columnwidth]{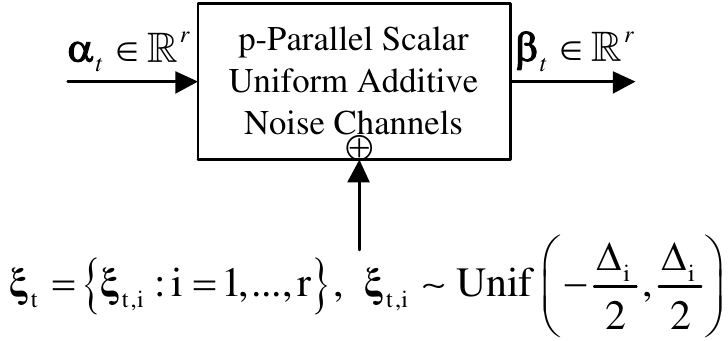}
\caption{An equivalent model to Fig.~\ref{fig:sdusq} based on uniform scalar additive noise channel.}
\label{fig:equivalent_additive_uniform_noise_channel}
\end{figure}
\par In view of the assumption that the uniform scalar quantizer with subtractive dither operates using memoryless entropy coding over time, the following theorem holds.
\begin{theorem}(Upper bound {to Problem \ref{problem1}})\label{theorem:general_bound}{\ \\}
Consider the realization of the zero-delay source-coding scheme illustrated in Fig.~\ref{fig:noisy_communication_system} with the ${r}-$ parallel $\awgn$ channels replaced by ${r}-$parallel independently operating $\sdusq$ illustrated in Fig. \ref{fig:replacement}. If the vector process $\{\tilde{\bm  \beta}_t:~t\in\mathbb{N}_0\}$ of the quantized output is jointly entropy coded conditioned to the dither signal values over spatial components at each time step $t$, then, the operational zero-delay rate, $R^{\op}_{\zd}(D)$, satisfies the following upper bound:
\begin{align}
R^{\op}_{\zd}(D)\leq{R}_{{\text{GM}}}^{\na}(D)+\frac{{r}}{2}\log_2\left(\frac{\pi{e}}{6}\right)+1,~{r\leq{p}},\label{operational_upper_bound:1}
\end{align}
while the $\mse$ distortion achieves the end-to-end  distortion $D$ of the system.
\end{theorem}
\begin{proof}
See Appendix~\ref{proof:theorem:general_bound}.
\end{proof}
\par {Recall that Theorem \ref{theorem:general_bound} provides an upper bound to Problem \ref{problem1} by taking into account the fact that ${R}_{{\text{GM}}}^{\na}(D)$ is characterized by a reverse-waterfilling in dimension. If $r=p$, then $\rank({\tilde{H}})$ is full rank, whereas if $r<p$, then, $\rank({\tilde{H}})$ is rank deficient and to compute we ``remove'' the $p-r$ dimensions with zero rate.}
\par {The previous main result combined with the lower bound of \S\ref{sec:new_realization_scheme}, leads to the following corollary.}
{
\begin{corollary}(General Bounds on Problem \ref{problem1})\label{corollary:bounds_upper_lower}{\ \\}
Consider the coding scheme of \S\ref{sec:upper_bounds_zero_delay_codes}. Then, for $\mathbb{R}^p-$valued Gauss-Markov sources the following bounds hold:
\begin{align}
{R}_{{\text{GM}}}^{\na}(D)\leq{R}^{\op}_{\zd}(D)\leq{R}_{{\text{GM}}}^{\na}(D)+\frac{{r}}{2}\log\left(\frac{\pi{e}}{6}\right)+1,\label{general_bounds:eq.1}
\end{align}
where $r\leq{p}$. 
\end{corollary}
\begin{proof}
This is obtained using Theorem \ref{theorem:bounds} (restricted to the class of Gaussian sources of Problem \ref{problem1}), and Theorem \ref{theorem:general_bound}.
\end{proof}
}
\par In the next remark we explain how the bounds derived in \eqref{general_bounds:eq.1} apply to vector Gauss $\ar$ sources of any order provided, the distortion function is chosen appropriately. {This result generalize the $\mathbb{R}^p-$valued Gauss-Markov source model of Problem \ref{problem1} to any vector-valued Gauss-Markov model of order higher than one.}

\begin{remark}(Generalization)\label{theorem:any_order}{\ \\}
Consider a vector Gaussian $\ar(s)$ process, in state space representation,  where $s$ is a positive integer.
\begin{align}
{\bf x}_{t+1}=\sum_{j=1}^sA_j{\bf x}_{t-j+1}+B{\bf w}_t,\label{vector_ar-p}
\end{align}
where $A_j\in\mathbb{R}^{p\times{p}}$, and $B\in\mathbb{R}^{{p}\times{q}}$ are deterministic matrices, ${\bf x}_0\in\mathbb{R}^{p}\sim{\cal N}(0;\Sigma_{{\bf x}_0})$, and ${\bf w}_t\in\mathbb{R}^q\sim{\cal N}(0;I)$ is an $\IID$ Gaussian sequence, independent of ${\bf x}_0$. Clearly, for $s=1$, \eqref{vector_ar-p} gives as a special case the source model described by \eqref{state_space_model}.\\
Then, \eqref{vector_ar-p} can be expressed as an augmented vector-valued Gauss-Markov process as follows:
\begin{align}
\tilde{\bf x}_{t+1}=\tilde{A}\tilde{\bf x}_{t}+\tilde{B}\tilde{\bf w}_t,\label{vector_ar1}
\end{align}
where $\tilde{A}\in\mathbb{R}^{sp\times{sp}}$, and $\tilde{B}\in\mathbb{R}^{sp\times{sq}}$ are deterministic matrices, $\tilde{\bf x}_0\in\mathbb{R}^{sp}\sim{\cal N}(0;\Sigma_{\tilde{\bf x}_0})$, and $\tilde{\bf w}_t\in\mathbb{R}^{sq}\sim{\cal N}(0;\Sigma_{\tilde{\bf w}_t})$ is an $\IID$ Gaussian sequence, independent of $\tilde{\bf x}_0$. \eqref{vector_ar1} is precisely the source model described by \eqref{state_space_model} under the assumption of full observation of the ``new'' source vector $\tilde{\bf x}_{t+1}\in\mathbb{R}^{sp}$.
\end{remark}
\begin{proof}
See Appendix~\ref{proof:theorem:any_order}.
\end{proof}

\par Clearly, Theorem~\ref{theorem:any_order}  can also be used to the special case where one wants to reformulate a {\it scalar}, i.e., $p=1$, Gaussian $\ar(s)$ process, where $s$ is a positive integer, into a {\it $\mathbb{R}^s$-valued} Gauss-Markov process. 

\subsubsection*{\bf Connections to existing bounds in the literature}\label{subsection:special_cases}

\par In the next remark, we use \eqref{general_bounds:eq.1} to draw connections to existing results in the literature.

\begin{remark}(Special cases)\label{remark:discussion}
\begin{itemize}
\item[(1)] For {\it scalar Gaussian $\ar$ sources}, i.e., {${r}=p=1$ ($\tilde{H}$ is always full rank in this special case, i.e., no reverse-waterfilling kicks in)}, \eqref{operational_upper_bound:1} degenerates to
\begin{align}
R^{\op}_{\zd}(D)\leq{R}_{{\text{GM}}}^{\na}(D)+\frac{1}{2}\log_2\left(\frac{\pi{e}}{6}\right)+1,\label{operational_upper_bound:scalar}
\end{align}
{where ${R}_{{\text{GM}}}^{\na}(D)$ is given by $\eqref{optimal_solution:scalar:eq.1}$.} This result coincides with the bound obtained in \cite[Theorem 7]{derpich:2012}. However, the upper bound in \cite{derpich:2012} is obtained using a realization scheme with four filters instead of only one that we use in our scheme. Moreover, our result holds for unstable Gaussian sources too.
\item[(2)] Compare to \cite{silva:2011}, we use $\ecdq$ based on a different realization setup (without control signals) that results into obtaining different lower and upper bounds. Moreover, \cite{silva:2011} provides an upper bound which holds only for {\it scalar} Gaussian $\ar$ sources (see \cite[Equation (22)]{silva:2011}). 
\end{itemize}
\end{remark} 

\subsection{Vector Quantization}\label{subsection:vector_quantization}

\par It is interesting to observe that if instead of $\sdusq$ we quantize over a vector lattice quantizer followed by memoryless entropy coded conditioned to the dither, then, the upper bound in \eqref{operational_upper_bound:1} becomes
\begin{align}
R^{\op}_{\zd}(D)\leq{R}_{{\text{GM}}}^{\na}(D)+\frac{{r}}{2}\log_2\left({2\pi{e}G_{{r}}}\right)+1,~{r\leq{p}},\label{operational_upper_bound_vector_quantization:1}
\end{align}
where $G_{{r}}$ is the normalized second moment of the lattice \cite{zamir:2014} defined by $G_{{r}}=\frac{1}{{r}}\frac{\mathbb{E}\{||{\bf q}||^2\}}{{\vol}^{\frac{2}{{r}}}},~{r\leq{p}}$, and $\vol$ is the volume of the basic cell of a vector quantizer.
\par If we take the average rate normalized per {total number of dimensions ``p''}, then, \eqref{operational_upper_bound_vector_quantization:1} becomes
\begin{align}
\frac{1}{p}R^{\op}_{\zd}(D)\leq\frac{1}{p}{R}_{\text{GM}}^{\na}(D)+\frac{{r}}{2p}\log_2\left({2\pi{e}G_{{r}}}\right)+\frac{1}{p},~{r\leq{p}}.\label{operational_upper_bound_vector_quantization:111}
\end{align}
Additionally, by assuming {$p\longrightarrow\infty$, i.e., infinite dimensional vector-valued Gaussian source}, then, we obtain
\begin{align}
\lim_{p\rightarrow\infty}\frac{1}{p}R^{\op}_{\zd}(D)\leq\lim_{p\rightarrow\infty}\frac{1}{p}{R}_{\text{GM}}^{\na}(D),\label{operational_upper_bound_vector_quantization:2}
\end{align} 
{because even if $r=p$ we know from \cite[Corollary 7.2.1]{zamir:2014}, that $G_p\rightarrow\frac{1}{2\pi{e}}$ which still cancels the second right hand side term in \eqref{operational_upper_bound_vector_quantization:111}.}
Interestingly, using vector quantization for infinite dimensional vector Gaussian source, the term due to space-filling loss and the loss due to entropy coding asymptotically goes to zero. Hence, utilizing {the latter and Theorem \ref{theorem:bounds}} we obtain
\begin{align}
\lim_{p\rightarrow\infty}\frac{1}{p}{R}_{\text{GM}}^{\na}(D)\leq\lim_{p\rightarrow\infty}\frac{1}{p}R^{\op}_{\zd}(D)\leq\lim_{p\rightarrow\infty}\frac{1}{p}{R}_{\text{GM}}^{\na}(D),\label{tight_bound}
\end{align}
i.e., ${R}_{\text{GM}}^{\na}(D)$ approximates {$R^{\op}_{\zd}(D)$ of Problem \ref{problem1}}.

\par In the next remark, we summarize two important observations that arise as an outcome of the previous result.

\begin{remark}(Improvements on bounds of {Theorem \ref{theorem:bounds}}){\ \\}
{The observation in \S\ref{subsection:vector_quantization}, \eqref{tight_bound} implies that in Theorem \ref{theorem:bounds}, inequalities ${R}^{\na}(D)\leq\widehat{R}^{\na}(D)\leq{R}^{\op}_c(D)\leq{R}_{\zd}^{\op}(D)$ hold with equality if one assumes infinite dimensional Gauss-Markov sources subject to an asymptotic $\mse$ distortion.} The latter also holds for scalar-valued Gaussian $\ar$ sources of an infinite order because this class of sources can be modified to infinite dimensional vector-valued Gauss-Markov sources using a trivial extension of Remark \ref{theorem:any_order}.
\end{remark}

\section{Examples}\label{sect:examples}

\par In this section, we present examples to demonstrate the validity of our theoretical framework of \S\ref{sec:new_realization_scheme}, \S\ref{sec:main_results}. In our numerical experiments, we also compute $R^{\op}_c(D)$ since the desired class of zero-delay codes, i.e., $R^{\op}_{\zd}(D)$, is a subclass of the general class of causal codes. We note that in each of the following examples, the corresponding lower and upper bounds to {Problem 1} are evaluated by invoking the $\sdp$ solver operating in the CVX platform \cite{cvx}. The operational rates are computed based on a $\sdusq$, however, in one example we consider a vector quantizer, namely, a $D_4$-lattice which is a 4-dimensional lattice \cite{conway-sloane1999}.  

\subsection{Stable Gaussian $\ar$ Sources}\label{sub:stable}

\par In what follows we discuss two examples of stable Gaussian $\ar$ sources.

\begin{example}($\mathbb{R}^4$-valued stable Gauss-Markov source){\ \\}\label{example2}
We consider a $\mathbb{R}^4$-valued Gauss-Markov source given  by
\begin{align}
{\bf x}_{t+1} 
=\underbrace{\begin{bmatrix}
0.0551 &   0.0893  &  0.0051  &  0.0649\\
0.0708  &  0.0896 &   0.0441  &  0.0278\\
0.0291   & 0.0126  &  0.0030  &  0.0676\\
0.0511   & 0.0207   & 0.0457  &  0.0591
 \end{bmatrix}}_{A}
 {\bf x}_{t} 
+{B}{\bf w}_{t},\label{nonstationary:two:dimension:equation2}
\end{align}
where $B{\bf w}_t\sim{\cal N}(0;\Sigma_{{\bf w}})$, $\Sigma_{{\bf w}}\sim{\cal N}(0;I)$. This example corresponds to \eqref{state_space_model} for $p=q=4$. The eigenvalues of matrix $A$ are {$(\mu_{A,1},\mu_{A,2},\mu_{A,3},\mu_{A,4})=(0.1953, -0.0383, -0.0045, 0.0542)$. Since $|\mu_{A,i}|<1,~\forall{i},$ then, the source is stable.} By invoking in the $\sdp$ solver to the semidefinite representation of $R_{\text{GM}}^{\na}(D)$ given by \eqref{eq:semidefinite_representation1}, we plot the theoretical lower and upper bounds on $R_{c}^{\op}(D)$ and $R_{\zd}^{\op}(D)$. The actual achievable rates are implemented using a $\sdusq$ followed by Huffman entropy coding that corresponds to the scheme of Fig. \ref{fig:replacement}. The rate-distortion curves are illustrated in Fig. \ref{fig:bounds:stable:dim4:SDUSQ}. 

\begin{figure}[htp]
\centering
\includegraphics[width=0.7\columnwidth]{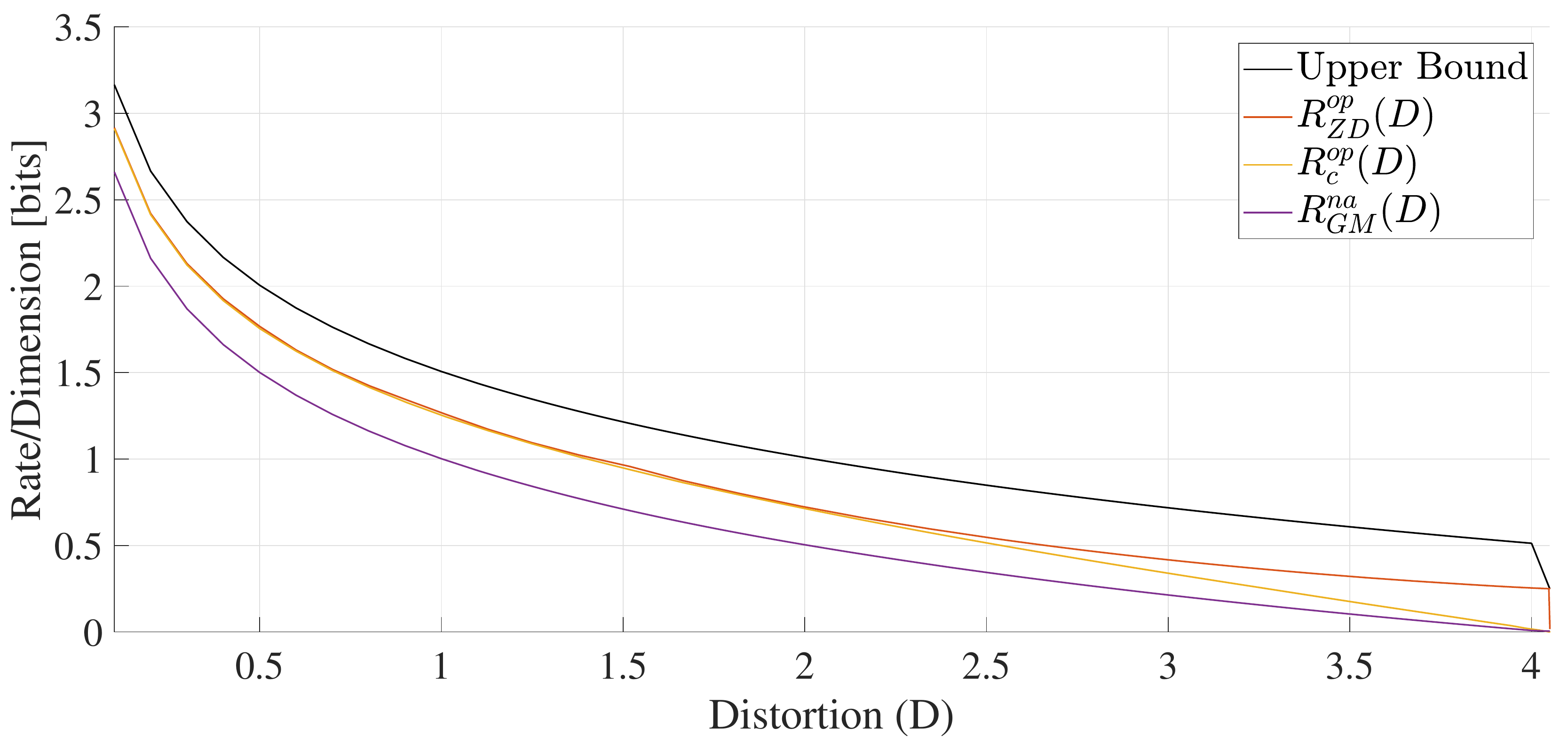}
\caption{Theoretical bounds on $R_{\zd}^{\op}(D)$ and achievable rates of $R_{c}^{\op}(D)$ and $R_{\zd}^{\op}(D)$ normalized per dimension that correspond to a $\mathbb{R}^4$-valued stable Gauss-Markov source.}
\label{fig:bounds:stable:dim4:SDUSQ}
\end{figure}
{For this example, the reverse-waterfilling kicks in when $D>3.95$. This means that the number of active dimension is $r<4$. Off course, when $D=D_{\max}$, then, $r=0$ because all dimensions are inactive.}
\end{example}

\begin{example}(Scalar-valued stable Gaussian $\ar(2)$ source){\ \\}\label{example2}
We consider a scalar Gaussian $\ar(2)$ source\footnote{The Gaussian source of \eqref{scalar_stable_example:ar2} is stable because for $a=0.3$ and $b=0.5$ it satisfies $a+b<1$ and $b-a<1$ and $|b|<1$ (see cf. \cite[Chapter 5]{shanmugan:1988}).} given by
\begin{align}
{\bf x}_{t+1}=0.3{\bf x}_{t}+0.5{\bf x}_{t-1}+{\bf w}_t, {\bf w}_t\sim{\cal N}(0;1).\label{scalar_stable_example:ar2}
\end{align}
Using Remark~\ref{theorem:any_order}, \eqref{scalar_stable_example:ar2} can be reformulated to a $\mathbb{R}^2$-valued Gauss-Markov source as follows:
\begin{align}
\underbrace{\begin{bmatrix}
   {\bf x}_{t+1}\\
   {\bf x}_{t} 
 \end{bmatrix}}_{\tilde{\bf x}_{t+1}} 
=\underbrace{\begin{bmatrix}
    0.3 & 0.5 \\
    1 & 0 
 \end{bmatrix}}_{A}
\underbrace{\begin{bmatrix}
   {\bf x}_{t}\\
   {\bf x}_{t-1} 
 \end{bmatrix}}_{\tilde{\bf x}_{t}}  
+\underbrace{\begin{bmatrix}
    1 & 0 \\
    0 & 0 
 \end{bmatrix}}_{B} 
\underbrace{\begin{bmatrix}
{\bf w}_{t} \\
0
\end{bmatrix}}_{\tilde{\bf w}_t},\label{modified:var1}
\end{align}
where $B\tilde{\bf w}_t\sim{\cal N}(0;\begin{bmatrix}
1 & 0\\
0 & 0
\end{bmatrix})$.
The modified Gaussian source corresponds to \eqref{state_space_model} for $p=q=2$. Since the original source is stable, the modified source is also stable. This can be verified by {observing that the eigenvalues of the augmented matrix $A$ are $(\mu_{A,1},~\mu_{A,2})=(0.8728, -0.5728)$. Since $|\mu_{A,i}|<1,~\forall{i}$, the source is stable as expected.} By invoking the $\sdp$ solver to the semidefinite representation of $R_{\text{GM}}^{\na}(D)$ given by \eqref{eq:semidefinite_representation2}, we plot the theoretical lower and upper bounds to $R_{c}^{\op}(D)$ and $R_{\zd}^{\op}(D)$ whereas the actual achievable rates are implemented using a $\sdusq$ followed by Huffman entropy coding that corresponds to the scheme of Fig. \ref{fig:replacement}. The rate-distortion curves are illustrated in Fig. \ref{fig:bounds:stable:dim2:ar2}. 

\begin{figure}[htp]
\centering
\includegraphics[width=0.7\columnwidth]{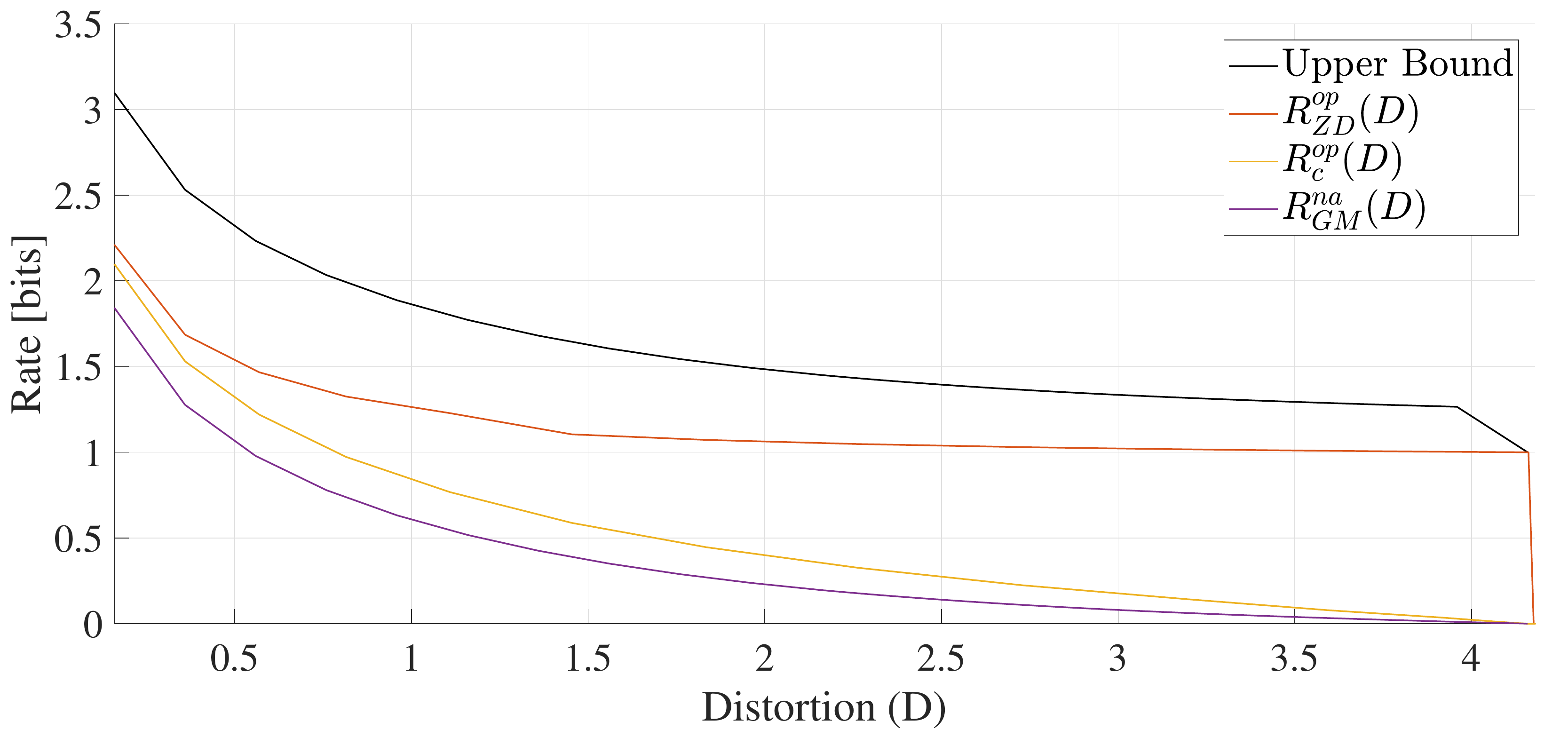}
\caption{Theoretical bounds on $R_{\zd}^{\op}(D)$ and achievable rates of $R_{c}^{\op}(D)$ and $R_{\zd}^{\op}(D)$ normalized per dimension that corresponds to a scalar-valued stable Gaussian $\ar(2)$ source.}
\label{fig:bounds:stable:dim2:ar2}
\end{figure}
{It is interesting to point out that the augmented state space representation of \eqref{modified:var1} does not change the number of active dimensions. Instead, for $D\in(0,D_{\max}]$, only the first dimension is active ($r=1$) whereas $p=2$. This is expected for two reasons; first the original source \eqref{scalar_stable_example:ar2} is scalar and the reverse-waterfilling is reasonable to affect such source, and second, the second dimension of the augmented representation yields a ``noiseless'' source which in turn means that it can be estimated without using a budget of the distortion. Again, when $D=D_{\max}$, then, $r=0$ because all dimensions are inactive.}
\end{example}

\subsection{Unstable Gaussian $\ar$ Sources}\label{sub:unstable}

\par In what follows we discuss two examples of unstable Gaussian $\ar$ sources.

\begin{example}($\mathbb{R}^4$-valued unstable Gauss-Markov source){\ \\}\label{example4}
We consider a $\mathbb{R}^4$-valued Gauss-Markov source given by
\begin{align}
{\bf x}_{t+1} 
=\underbrace{\begin{bmatrix}
0.8147  &  0.6324  &  0.9575  &  0.9572 \\
    0.9058 &   0.0975  &  0.9649 &   0.4854\\
    0.1270  &  0.2785  &  0.1576  &  0.8003\\
    0.9134   & 0.5469  &  0.9706   & 0.1419
 \end{bmatrix}}_{A}
 {\bf x}_{t} 
+{B}{\bf w}_{t},
  \label{unstable:source:dim_4}
\end{align}
where $B{\bf w}_t\sim{\cal N}(0;\Sigma_{{\bf w}})$, $\Sigma_{{\bf w}}\sim{\cal N}(0;I)$. This example corresponds to \eqref{state_space_model} for $p=q=4$. The eigenvalues of matrix $A$ are {$(\mu_{A,1},\mu_{A,2},\mu_{A,3},\mu_{A,4})=(2.4022, -0.0346, -0.7158, -0.4400)$. Since $|\mu_{A,1}|>1$, then, the source is unstable.} Hence, using \eqref{unstable_eig}, we require rates (normalized per dimension):
\begin{align}
\frac{1}{4}{R}_{\text{GM}}^{\na}(D)\geq\frac{1}{4}\log|\mu_{A,1}|=0.3161~\text{bits}.\label{sum_log_dim4}
\end{align} 
By invoking the $\sdp$ solver to the semidefinite representation of $R_{\text{GM}}^{\na}(D)$ given by \eqref{eq:semidefinite_representation1}, we plot the theoretical lower and upper bounds on $R_{c}^{\op}(D)$ and $R_{\zd}^{\op}(D)$ and the bound of \eqref{sum_log_dim4}. The actual achievable rates are implemented using a $D_4$-lattice followed by Huffman entropy coding that corresponds to the scheme of Fig. \ref{fig:replacement}. The $D_4$-lattice \cite{conway-sloane1999} has a normalized second moment $G_4=0.076603$. Hence, by substituting this value in \eqref{operational_upper_bound_vector_quantization:1} (normalized per dimension), we obtain:
\begin{align}
\frac{1}{4}{R}_{\text{GM}}^{\na}(D)&\leq\frac{1}{4}R^{\op}_{\zd}(D)\leq\frac{1}{4}{R}_{\text{GM}}^{\na}(D)+0.4439~\text{bits}.\label{example:unstable:dim4_lower}
\end{align}
The rate-distortion curves, for the distortion region $D\in(0,3]$, are illustrated in Fig. \ref{fig:bounds:unstable:dim4:z4}. 

\begin{figure}[htp]
\centering
\includegraphics[width=0.7\columnwidth]{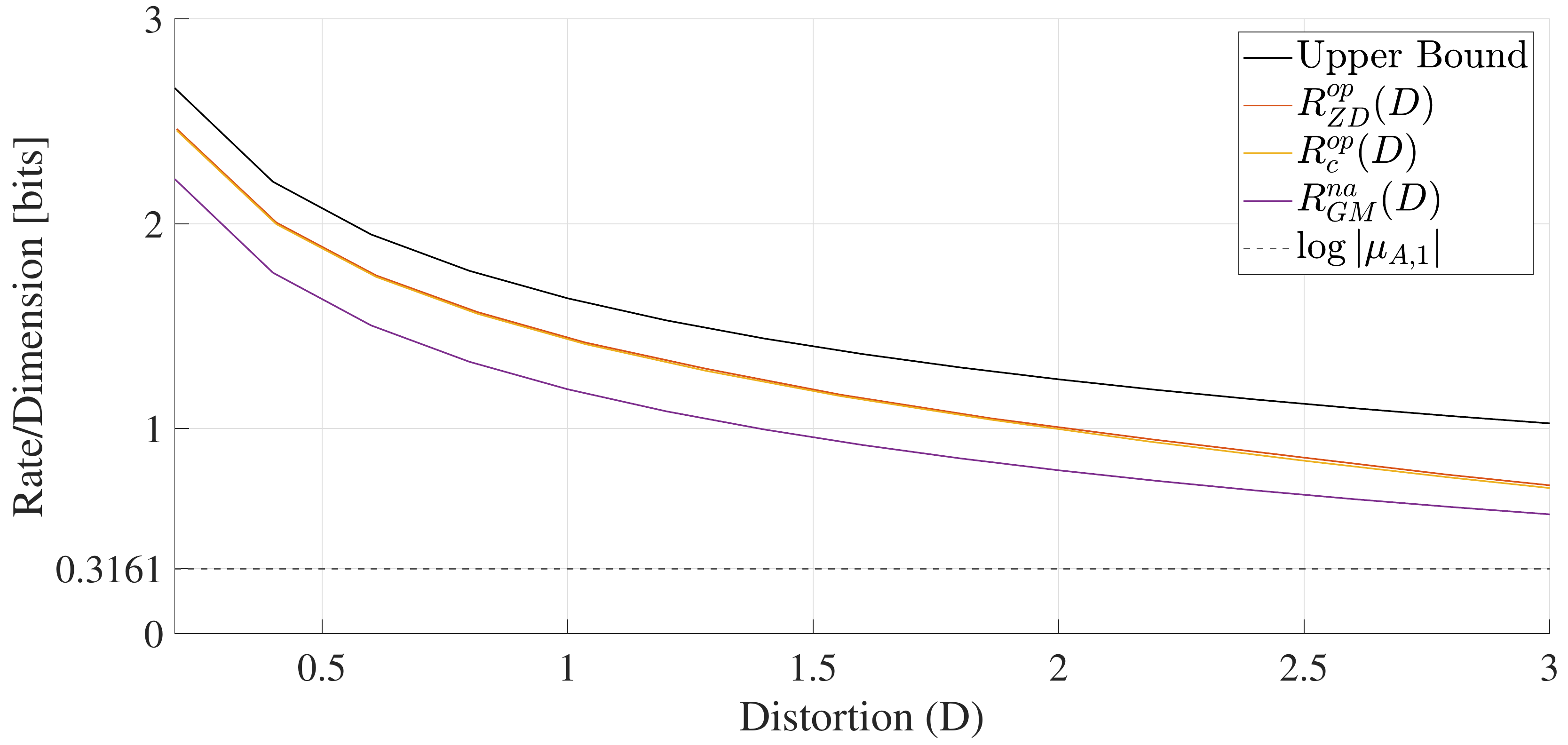}
\caption{Theoretical bounds on $R_{\zd}^{\op}(D)$ and achievable rates of $R_{c}^{\op}(D)$ and $R_{\zd}^{\op}(D)$ normalized per dimension that correspond to an $\mathbb{R}^4$-valued unstable Gauss-Markov source.}
\label{fig:bounds:unstable:dim4:z4}
\end{figure}
{We note that $D_4$-lattice is appropriate to use as long as $r=p$. If however the reverse-waterfilling kicks in then one has to use the vector quantizer that matches the active number of dimensions, i.e., $r\in\{1,2,3\}$. It should be remarked that, in general, for unstable sources it is expected that $r\neq{0}$ because at least for dimensions with $|\mu_{A,i}|>1$, for some $i$, $D_{\max}$ is infinite.}
\end{example}

\begin{example}(Scalar-valued unstable Gaussian $\ar(2)$ source){\ \\}\label{example5}
We consider a scalar Gaussian $\ar(2)$ source given by
\begin{align}
{\bf x}_{t+1}=1.2{\bf x}_{t}+0.5{\bf x}_{t-1}+{\bf w}_t,~{\bf w}_t\sim{\cal N}(0;1).\label{scalar_unstable_example:ar2}
\end{align}
Using Remark~\ref{theorem:any_order}, \eqref{scalar_unstable_example:ar2} can be reformulated to a $\mathbb{R}^2$-valued Gauss-Markov source as follows:
\begin{align}
\tilde{\bf x}_{t+1} 
=\underbrace{\begin{bmatrix}
    1.2 & 0.5 \\
    1 & 0 
 \end{bmatrix}}_{A}
\tilde{\bf x}_{t} 
+\underbrace{\begin{bmatrix}
    1 & 0 \\
    0 & 0 
 \end{bmatrix}}_{B} 
\tilde{\bf w}_{t},\label{modified:unstable:var1}
\end{align}
where $B\tilde{\bf w}_t\sim{\cal N}(0;\begin{bmatrix}
1 & 0\\
0 & 0
\end{bmatrix})$.
The modified Gaussian source now corresponds to \eqref{state_space_model} for $p=q=2$. Since the original source was unstable, the modified source is also unstable. This can also be verified by observing that the eigenvalues of the augmented matrix $A$ are {$(\mu_{A,1},\mu_{A,2})=(1.5274,-0.3274)$. Since $|\mu_{A,1}|>1$, then, the source is unstable.} 
Moreover, using \eqref{unstable_eig}, we require rates (normalized per dimension):
\begin{align}
{R}_{\text{GM}}^{\na}(D)\geq\log|\mu_{A,1}|=0.611~\text{bits}.\label{sum_log_dim2}
\end{align}
By invoking the $\sdp$ solver of the semidefinite representation of $R_{\text{GM}}^{\na}(D)$ given by \eqref{eq:semidefinite_representation2}, we plot the theoretical lower and upper bounds on $R_{c}^{\op}(D)$ and $R_{\zd}^{\op}(D)$ and the bound of \eqref{sum_log_dim2}. The actual achievable rates are implemented using a $\sdusq$ followed by Huffman entropy coding that corresponds to the scheme of Fig. \ref{fig:replacement}. The rate-distortion curves, {for a distortion region $D\in(0,3]$,} are illustrated in Fig. \ref{fig:bounds:unstable:dim2:ar2}. 

\begin{figure}[htp]
\centering
\includegraphics[width=0.7\columnwidth]{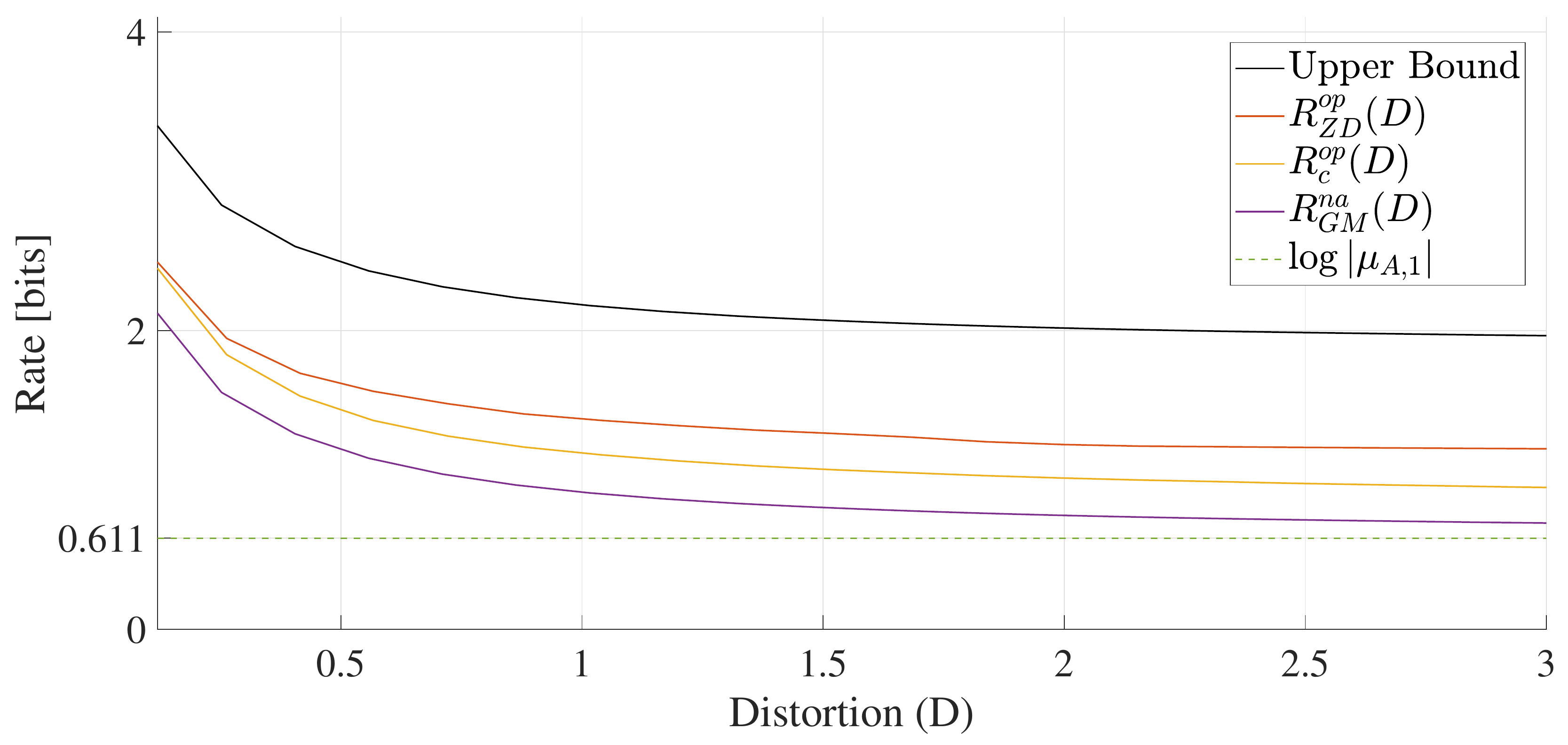}
\caption{Theoretical bounds on $R_{\zd}^{\op}(D)$ and achievable rates of $R_{c}^{\op}(D)$ and $R_{\zd}^{\op}(D)$ normalized per dimension that corresponds to a scalar unstable Gaussian $\ar(2)$ source.}
\label{fig:bounds:unstable:dim2:ar2}
\end{figure}
{Similar to Example \ref{example2}, the augmented state space representation of \eqref{modified:unstable:var1} does not alternate the original source model of \eqref{scalar_unstable_example:ar2} because for the whole distortion region of $D\in(0,D_{\max}]$, $r=1$ (only the first dimension is active).}
\end{example}

\subsection{Discussion of the results}\label{subsec:discussion}
From the rate-distortion performances shown in Fig.~\ref{fig:bounds:stable:dim4:SDUSQ} - \ref{fig:bounds:unstable:dim2:ar2}, it is clear that the proposed feedback-quantization based scheme is able to get very close to the $\opta$.  At high rates, it can be observed that the theoretical upper bound tends to be approximately $\frac{1}{r}$~bits above the operational rates. This is expected since at high rates the loss of the entropy coder vanishes and the only excess rate will be due to that of the space-filling loss of the quantizer \cite{linder-zamir2006}. At the other extreme, when the bit rates are very small, the impact due to the space-filling loss will vanish but the excess rate of the entropy coder will increase. To reduce the excess rate of the entropy coder, one could either apply fixed-rate coding instead of variable rate coding, or avoid the use of prefix codes in the entropy coder as suggested for instance in \cite{szpankowski:2011}. The latter approach assumes that individual codewords can be distinguished by other means than using a prefix.

%
%
%
%

\section{Conclusions and Future directions}\label{sec:conclusions}
\paragraph*{Conclusions} In this work, we considered zero-delay source coding of a vector-valued Gauss-Markov source subject to a $\mse$ fidelity criterion. We derived lower and upper bounds to the operational zero-delay Gaussian $\rdf$ by means of a new feedback realization scheme which is characterized by a resource allocation problem in spatial dimension and predictive coding. 
The performance of this scheme is analyzed when using scalar and vector lattice quantization followed by memoryless entropy coding. For infinite dimensional Gauss-Markov sources, we observed that the $\nrdf$ approximates the operational zero-delay $\rdf$. Our theoretical results are corroborated via various simulation examples.
\paragraph*{Future Directions} The results derived in this paper makes it interesting to consider the extension of this framework to closed-loop control systems where the goal is to identify the fundamental performance limitations of controlled processes under communication constraints. {The extension of this framework to partially observable Gaussian sources with various channel parameters similar to \cite[Chapter 11]{matveev-savkin:2009} is also under consideration.} 


%
%
%
%
%
%

\appendices

\section{Proof of Lemma~\ref{theorem:solution:semidefinite}}\label{proof:theorem:solution:semidefinite}

{\bf(2)} First, we note that \eqref{alternative:optimization:inf} can be reformulated as follows
\begin{align}
R_{\text{GM}}^{\na}(D)&=\min_{\substack{0\prec\Pi\preceq\Lambda \\ \trace{(\Pi)}\leq{D}}}\frac{1}{2}\log|\Pi+A^{-1}BB{\T}(A{\T})^{-1}|\nonumber\\
&-\frac{1}{2}\log|\Pi|+\log\abs(|A|).\label{reformulate_1}
\end{align}
However, \eqref{reformulate_1} can be written as
\begin{align}
&R_{\text{GM}}^{\na}(D)\nonumber\\
&=\min_{\substack{\Pi\preceq\Lambda\\ \trace{(\Pi)}\leq{D}}}\frac{1}{2}\log|I+\Pi^{-\frac{1}{2}}A^{-1}BB{\T}(A{\T})^{-1}\Pi^{-\frac{1}{2}}|\nonumber\\
&+\log\abs(|A|)\label{reformulate_2}\\
&\stackrel{(a)}=\min_{\substack{\Pi\preceq\Lambda\\ \trace{(\Pi)}\leq{D}}}\frac{1}{2}\log|I+B{\T}(A{\T})^{-1}\Pi^{-1}A^{-1}B|\nonumber\\
&+\log\abs(|A|),\label{reformulate_3}
\end{align}
where $(a)$ follows from Sylvester's determinant identity \cite[Corollary 18.1.2]{harville:1997}. Next, we introduce a decision variable $Q_2\triangleq{I}+B{\T}(A{\T})^{-1}\Pi^{-1}A^{-1}B$. The monotonicity of the determinant function, results into writing \eqref{reformulate_2} as
\begin{align}
\min_{\substack{0\prec\Pi\preceq\Lambda\\ \trace(\Pi)\leq{D}\\ 0\prec Q_2 \preceq ({I}+B{\T}(A{\T})^{-1}\Pi^{-1}A^{-1}B)^{-1}}}\log\abs(|A|)-\frac{1}{2}\log|Q_2|.\label{monotonicity_2}
\end{align} 
Applying the Woodbury matrix identity \cite[Theorem 18.2.8]{harville:1997} in the inequality constraint $0\prec Q_2 \preceq ({I}+B{\T}(A{\T})^{-1}\Pi^{-1}A^{-1}B)^{-1}$, we obtain 
\begin{align}
0\prec Q_2 \preceq I-B{\T}(A\Pi{A}{\T}+BB{\T})^{-1}B. \label{woodbury_identity_2}
\end{align}
{Using \eqref{scalings:3:inf} in \eqref{woodbury_identity_2}, we can write the latter as} the linear matrix inequality condition of \eqref{lmi:2}.
This completes the proof.\qed

\section{Proof of Theorem \ref{theorem:alter:real:waterfilling}}\label{proof:theorem:alter:real:waterfilling}

{The proof depends on the joint diagonalization of the pair of symmetric positive definite matrices $(\Pi,~\Lambda)$. This is a version of the cogredient diagonalization approach derived in \cite[Theorem 8.3.1]{bernstein2011}.    \\
Consider the ordered pair ($\Pi$, $\Lambda$) where $\Pi\in\mathbb{R}^{p\times{p}}$ and $\Lambda\in\mathbb{R}^{p\times{p}}$. Denote the eigenvalue decomposition of $\Pi$ by:
\begin{align}
\Pi=U\T\tilde{\Pi}{U},\label{appendix:eig1}
\end{align}
where $U\in\mathbb{R}^{p\times{p}}$ is an orthogonal matrix, $\tilde{\Pi}=\diag(\mu_{\Pi,i})$ and $\mu_{\Pi,1}\geq\mu_{\Pi,2}\geq\ldots\geq\mu_{\Pi,p}$. Observe that from \eqref{appendix:eig1}, the square root $\Pi^{\frac{1}{2}}=U\T\tilde{\Lambda}^{\frac{1}{2}}{U}$. Also, denote the eigenvalue decomposition of $\Pi^{-\frac{1}{2}}\Lambda\Pi^{-\frac{1}{2}}$ by:
\begin{align}
\Pi^{-\frac{1}{2}}\Lambda\Pi^{-\frac{1}{2}}=V\T{S}{V}.\label{appendix:eig2}
\end{align}
where such that $V\in\mathbb{R}^{p\times{p}}$ is an orthogonal matrix, $S=\diag(\mu_{\Pi^{-\frac{1}{2}}\Lambda\Pi^{-\frac{1}{2}},i})\equiv\diag(\mu_{S,i})$ and $\mu_{S,1}\geq\mu_{S,2}\geq\ldots\geq\mu_{S,p}$. \\
The {\it simultaneous diagonalizer} $E\in\mathbb{R}^{p\times{p}}$ of the pair ($\Pi$, $\Lambda$) is defined as $E\triangleq\tilde{\Pi}^{\frac{1}{2}}V\Pi^{-\frac{1}{2}}$. Then, each matrix in the pair ($\Pi$, $\Lambda$) is diagonalized as in \eqref{diagonalized_lambdas_deltas}
where the diagonal matrix $\tilde{\Lambda}=\diag(\mu_{{\Lambda},i})$ is defined as $\tilde{\Lambda}\triangleq\tilde{\Pi}S$. Moreover, using \eqref{diagonalized_lambdas_deltas} then from \eqref{scalings:1:inf} and \eqref{scalings:2:inf} we obtain
\begin{align}
H=E^{-1}\tilde{H}E,~\Sigma_{\bf v}=E^{-1}\tilde{\Sigma}_{\bf v}(E\T)^{-1},\label{appendix:eig3}
\end{align}
where $\tilde{H}$ is the diagonal matrix given in \eqref{design_parameters_no_waterfilling2} and $\tilde{\Sigma}_{\bf v}$ is the diagonal matrix given in \eqref{scaling_no_waterfilling1}. The diagonal matrix $\tilde{H}$ can be further decomposed to the diagonal matrices $\Theta$ and $\Phi$ chosen as in \eqref{scaling_no_waterfilling1} and \eqref{scaling_no_waterfilling2}, respectively. The choice of $\Theta,~\Phi$ is not random but depends on the realization scheme of \eqref{optimal_realization_asymptotic}. $\Theta={\tilde{\Sigma}_{\bf v}}^{\frac{1}{2}}$ because the independent Gaussian noise process ${\bf v}_t\sim{\cal N}(0;\Pi{H\T})$ of \eqref{optimal_realization_asymptotic} can be equivalently reformulated using \eqref{appendix:eig3} to $E^{-1}\Theta{\bf v}_t\sim{\cal N}(0;I)$. Using the latter transformation and \eqref{appendix:eig3}, the realization of \eqref{optimal_realization_asymptotic} can be reformulated as in \eqref{optimal:realization:diagonal}.  $a)$ If $\tilde{H}\succ{0}$, then, $\forall{i}=1,\ldots,p$, $\mu_{\Pi,i}<\mu_{\Lambda,i}$, i.e., all $p$-dimensions of the scheme are active (no resource allocation). Hence $\tilde{H}$ is full rank, i.e., $r=\rank{(\tilde{H})}=p$.  $b)$ If $\tilde{H}\succeq{0}$, then, for some $i\in\{1,\ldots,p\}$, $\mu_{\Pi,i}=\mu_{\Lambda,i}$, i.e., some $p-r$ dimensions are inactive and for these dimensions the rate is zero. This in turn means that the data rate budget of the system is experiencing a resource allocation. Hence $\tilde{H}$ is rank deficient, i.e., $r=\rank{(\tilde{H})}<p$. This completes the proof. \qed


\section{Proof of Theorem~\ref{theorem:general_bound}}\label{proof:theorem:general_bound}

{The coding scheme described in Fig. \ref{fig:replacement} operates with an operational rate for each $t$ equal to the conditional entropy $H(\tilde{\bm \beta}_t|{\bf q}_t)$ where $\tilde{\bm \beta}_{t}=\{\tilde{\bm \beta}_{t,1},\ldots,\tilde{\bm \beta}_{t,r}\}$, $\tilde{\bm \beta}_{t,i}=Q_{\Delta_i}({\bm \alpha}_{t,i}+{\bf r}_{t,i}),~i=1,\ldots,r$, i.e., the entropy of the quantized output $\tilde{\bm \beta}_t$ conditioned on the $t-$value of the dither signal ${\bf q}_t$.} This leads to the following analysis:
\begin{align}
&\mathbb{H}(\tilde{\bm \beta}_t|{{\bf q}}_t)\stackrel{(a)}=I({\bm \alpha}_{t};{\bm \beta}_{t})\stackrel{(b)}=I({\bm \alpha}_{t};{\bm \alpha}_{t}+{\bm \xi}_{t})\nonumber\\
&={h}({\bm \alpha}_{t}+{\bm \xi}_{t})-h({\bm \xi}_{t})\nonumber\\
&\stackrel{(c)}={h}({\bm \alpha}\G_{t}+{\bf v}_{t})-h({\bf v}_{t})+\mathbb{D}({\bm \xi}_{t}||{\bf v}_t)-\mathbb{D}({\bm \alpha}_{t}+{\bm \xi}_{t}||{\bm \alpha}\G_{t}+{\bf v}_{t})\nonumber\\
&\stackrel{(d)}\leq{I}({\bm \alpha}\G_{t};{\bm \alpha}\G_{t}+{\bf v}_{t})+\mathbb{D}({\bm \xi}_{t}||{\bf v}_t)\nonumber\\
&\stackrel{(e)}=I({\bm \alpha}\G_{t};{\bm \alpha}\G_{t}+{\bf v}_{t})+\frac{{r}}{2}\log_2\left(\frac{\pi{e}}{6}\right)\nonumber\\
&=I({\bm \alpha}\G_{t};{\bm \beta}_t\G)+\frac{{r}}{2}\log_2\left(\frac{\pi{e}}{6}\right),~~{r\leq{p}},\label{proof:eq.1}
\end{align}
where $(a)$ stems from \cite[Theorem 5.2.1]{zamir:2014}; $(b)$ stems from the fact that the quantization noise is ${\bm \xi}_t={\bm \beta}_t-{\bm \alpha}_t$ (see Fig.~\ref{fig:equivalent_additive_uniform_noise_channel}); $(c)$ stems as follows; first note that the relative entropy between the distributions of two random vectors, a uniform random vector ${\bf n}$, and a Gaussian random vector ${\bf n}\G$,  both with the same moments of the quadratic form $\log_2{\bf P}_{{\bf n}\G}$, implies that $\mathbb{D}({\bf n}||{\bf n}\G)=h({\bf n}\G)-h({\bf n})$ (for details see, e.g., \cite[Proof of Theorem 8.6.5]{cover-thomas2006}). Then, using this property we can immediately deduce that $\mathbb{D}({\bm \xi}_{t}||{\bf v}_t)=h({\bf v}_t)-h({\bm \xi}_t)$ and $\mathbb{D}({\bm \alpha}_{t}+{\bm \xi}_{t}||{\bm \alpha}\G_{t}+{\bf v}_{t})=h({\bm \alpha}\G_{t}+{\bf v}_{t})-h({\bm \alpha}_{t}+{\bm \xi}_{t})$; $(d)$ stems from the fact that $D({\bm \alpha}_{t}+{\bm \xi}_{t}||{\bm \alpha}\G_{t}+{\bf v}_{t})\geq{0}$, with equality if and only if $\{{\bm \xi}_t:~t\in\mathbb{N}_0\}$ becomes a Gaussian distribution; $(e)$ follows from the fact that the differential entropy $h({\bf v}_t)$ of a Gaussian random vector with covariance $\Sigma_{\bf v}\triangleq\diag({\mu_{\Sigma_{\bf v},i}})$ is 
\begin{align*}
h({\bf v}_t)=\frac{1}{2}\log_2(2\pi{e})^{{r}}|\Sigma_{\bf v}|=\sum_{i=1}^{{r}}\frac{1}{2}\log_2(2\pi{e}){\mu_{\Sigma_{\bf v},i}},
\end{align*}
and the entropy $h({\bm \xi}_t)$ of the uniformly distributed random vector ${\bm \xi}_t=\{{\bm \xi}_{t,i}:~i=1,2,\ldots,{r}\}, {\bm \xi}_{t,i}\sim{\unif}\left(-\frac{\Delta_i}{2},\frac{\Delta_i}{2}\right)$ is $h({\bm \xi}_t)=\sum_{i=1}^{{r}}\frac{1}{2}\log_2\Delta^2_i,~{r\leq{p}}$. Since we have that ${\mu_{\Sigma_{\bf v},i}}=\frac{\Delta^2_i}{12}$, $i=1,\ldots,{r}$, the result follows.
\par Recall {from Corollary \ref{corollary:dpe} that 
\begin{align}
{I}({\bf x}^n;{\bf y}^n)=\sum_{t=0}^{n}{I}({\bm \alpha}_t;{\bm \beta}_t).\label{proof:eq.2}
\end{align}
}

\par Since we are assuming joint memoryless entropy coding of ${r}$ independently operating uniform scalar quantizers with subtractive dither, then, by \eqref{intantaneous_code_bounds}, for $t\in\mathbb{N}_0^n$, we obtain
\begin{align}
&\sum_{t=0}^{n}\mathbb{E}({\bf l}_t)\leq\sum_{t=0}^{n}\left(\mathbb{H}(\tilde{\bm \beta}_t|{{\bf q}}_t)+1\right)\nonumber\\
&\stackrel{(a)}\leq\sum_{t=0}^{n}\left(I({\bm \alpha}\G_{t};{\bm \beta}\G_{t})+\frac{{r}}{2}\log_2\left(\frac{\pi{e}}{6}\right)+1\right),~{r\leq{p}},\nonumber\\
&\stackrel{(b)}\leq{I}({\bf x}^{n,\text{G}};{\bf y}^{n,\text{G}})+(n+1)\frac{{r}}{2}\log_2\left(\frac{\pi{e}}{6}\right)+(n+1),\label{proof:eq.3}
\end{align}
where $(a)$ follows by \eqref{proof:eq.1} and $(b)$ follows from \eqref{proof:eq.2}.
\par {Then, by taking the  per unit time asymptotic limit of \eqref{proof:eq.3}} and then the infimum, we obtain 
\begin{align}
&\inf\limsup_{n\longrightarrow\infty}\frac{1}{n+1}\sum_{t=0}^{n}\mathbb{E}({\bf l}_t)\nonumber\\
&\leq\inf\limsup_{n\longrightarrow\infty}\frac{1}{n+1}{I}({\bf x}^{n,\text{G}};{\bf y}^{n,\text{G}})+\frac{{r}}{2}\log_2\left(\frac{\pi{e}}{6}\right)+1\nonumber\\
&\stackrel{(a)}\Longrightarrow{R}_{\zd}^{\op}\leq{\widehat{R}_{\text{GM}}^{\na}(D)}+\frac{{r}}{2}\log_2\left(\frac{\pi{e}}{6}\right)+1,\label{proof:eq.4}
\end{align}
where $(a)$ follows by \eqref{def:operational_zero_delay} and \eqref{infinite_time_nrdf_upper_bound} respectively, and ${\widehat{R}_{\text{GM}}^{\na}(D)}$ is the upper bound expression of ${R_{\text{GM}}^{\na}(D)}$ of the {$\mathbb{R}^r-$valued Gauss-Markov source modeled as in \eqref{state_space_model},~where $r\leq{p}$}.
\par {However, we know from the conditions of Theorem \ref{theorem:asymptotic_limit}, that the asymptotic limit ${R_{\text{GM}}^{\na}(D)}$ exists and is finite and the corresponding asymptotic values of its statistics exist and they are stationary.} Hence, we establish that ${\widehat{R}_{\text{GM}}^{\na}(D)=R_{\text{GM}}^{\na}(D)}$. This completes the proof.\qed

\section{Proof of Remark~\ref{theorem:any_order}}\label{proof:theorem:any_order}

\par The proof employs a simple augmentation of the state process. This is a standard method in state space models, see. e.g., \cite{caines1988}, {however, in what follows we sketch the proof for completeness}. In particular, observe that the state space model of \eqref{vector_ar-p} can be modified as follows:
{\allowdisplaybreaks
\begin{align}
\underbrace{\begin{bmatrix}
     {\bf x}_{t+1} \\ 
     {\bf x}_{t}\\
     \vdots\\
     {\bf x}_{t-s+2}
\end{bmatrix}}_{\tilde{\bf x}_{t+1}}&=
\underbrace{\begin{bmatrix}
      A_1 & {A}_2 &\ldots & A_{s-1}&A_s\\ 
      I & 0 & \ldots &0 & 0\\
      0& I & \ddots &0 & 0\\
      \vdots & \vdots&\ddots & \vdots & \vdots\\
      0 & 0 & \ldots &I & 0\\
\end{bmatrix}}_{\tilde{A}\in\mathbb{R}^{sp\times{sp}}}
\underbrace{\begin{bmatrix}
     {\bf x}_{t} \\ 
     {\bf x}_{t-1}\\
     \vdots\\
     {\bf x}_{t-s+1}
\end{bmatrix}}_{\tilde{\bf x}_t}\nonumber\\
&+\underbrace{\begin{bmatrix}
      B & 0 &\ldots & 0\\ 
      0 & 0 & \ldots & 0\\
     \vdots & \vdots & \ldots & \vdots\\
      0 & 0 & \ldots & 0\\
\end{bmatrix}}_{\tilde{B}\in\mathbb{R}^{sp\times{sq}}}
\underbrace{\begin{bmatrix}
     {\bf w}_{t} \\ 
     0\\
     \vdots\\
     0
\end{bmatrix}}_{\tilde{\bf w}_t}.\label{proof:aug.eq.3}
\end{align}
Thus, \eqref{proof:aug.eq.3} can be written in an augmented state space form as follows:
\begin{align}
\tilde{\bf x}_{t+1}=\tilde{A}\tilde{\bf x}_t+\tilde{B}\tilde{\bf w}_t. \label{proof:aug.eq.4}
\end{align}
}
The augmented state space representation in \eqref{proof:aug.eq.4} is a $\mathbb{R}^{sp}$-valued Gauss-Markov source under the assumption of full observation of the vector $\tilde{\bf x}_{t+1}\in\mathbb{R}^{sp}$.\qed 
%
%
%
%
%
%

\bibliographystyle{IEEEtran}
\bibliography{string,photis_quantization}

\end{document}